\documentclass[twocolumn]{autart}    

\usepackage{graphicx}          
\usepackage{color}

\usepackage{caption}

\usepackage{natbib}        
\usepackage{subfigure}
\usepackage{float}
\usepackage{algorithm}
\usepackage{amssymb}     
\usepackage{amsmath}
\usepackage{bm}
\usepackage{booktabs} 
\usepackage{etoolbox}
\usepackage{array}

\let\classAND\AND
\let\AND\relax
\usepackage{algorithmic}

\let\AND\classAND
\AtBeginEnvironment{algorithmic}{\let\AND\algoAND}

\newtheorem{scheme}{Scheme}

\newenvironment{proof}{\textbf{Proof:}}{\hfill  $\blacksquare$}

\newcommand{\tabincell}[2]{\begin{tabular}{@{}#1@{}}#2\end{tabular}}  

\usepackage{array}

\begin{document}

\begin{frontmatter}

\title{Secure Consensus with Distributed Detection \\ via Two-hop Communication\thanksref{footnoteinfo}} 

\thanks[footnoteinfo]{This work was supported in part by the JST 
	CREST Grant No.~JPMJCR15K3 and by JSPS under Grant-in-Aid for 
	Scientific Research Grant No.~18H01460.
    The material in this paper was partially presented at the 8th IFAC Workshop on Distributed Estimation and Control in Networked Systems (NecSys2019), September 16-17, 2019, Chicago, IL, USA.}

\author[Tokyo]{Liwei Yuan}\ead{yuan@sc.dis.titech.ac.jp},    
\author[Tokyo]{Hideaki Ishii}\ead{ishii@c.titech.ac.jp}               

\address[Tokyo]{Department of Computer Science, Tokyo Institute of Technology, 4259-J2-54, Yokohama 226-8502, Japan}


\begin{abstract}                          
In this paper, we consider a multi-agent resilient consensus problem,
where some of the nodes may behave maliciously. The approach
is to equip all nodes with a scheme to detect neighboring
nodes when they behave in an abnormal fashion. To this end,
the nodes exchange not only their own states but also information
regarding their neighbor nodes. Such two-hop communication has
long been studied in fault-tolerant algorithms in computer science. 
We propose two distributed schemes for detection
of malicious nodes and resilient consensus with different requirements
on resources for communication and the structures of the networks.
In particular, the detection schemes become effective under certain
connectivity properties in the network so that the non-malicious nodes
can share enough information about their neighbors. It is shown that
the requirements are however less stringent than those for conventional  
algorithms.
A numerical example is presented to demonstrate the performance of the  
proposed methods in wireless sensor networks.

\end{abstract}

\end{frontmatter}

\section{Introduction}

Recently, studies on cyber security issues of networked systems have 
gained much attention. With large-scale implementations of distributed, networked
applications such as clock synchronization (\cite{kadowaki2014event}), energy management (\cite{yang2013consensus}), 
formation control (\cite{dimarogonas2010stability}) and so on, consensus problems in the presence of adversary agents creating failures and attacks have become crucial; see, e.g.,
\cite{dibaji2015consensus, dibaji2018resilient, leblanc2013resilient, pasqualetti2012consensus}. 
Misbehaving nodes considered in existing 
works can be characterized by three adversary models according to the 
scopes of threat levels: Non-colluding/crash, malicious, and Byzantine 
models. Crash agents may cause random node failures, stopping the normal functionalities and not responding to other agents. In the malicious and Byzantine 
cases, misbehaving agents are capable to manipulate their data 
arbitrarily and may even collaborate with each other. 
Malicious nodes are limited as they must broadcast the same messages 
to all of their neighbors, while Byzantine nodes are more adversarial, being capable to send 
different messages (e.g., \cite{chandra1996unreliable, lamport1982byzantine, Lynch, teixeira2012attack}).

In the systems control literature, resilient consensus in 
the case of malicious model has been studied widely (\cite{dibaji2017resilient, leblanc2013resilient}). 
This model is suitable for multi-agent applications such as wireless
sensor networks and autonomous robotic networks, where the information
exchange among the nodes is via broadcasting and sensing.
Fault tolerant techniques using the so-called mean subsequence 
reduced (MSR) type algorithms have been found effective there and have also been used in computer science (e.g., \cite{azadmanesh2002asynchronous, mendes2015multidimensional, vaidya2012iterative}) and robotics (e.g., \cite{guerrero2017formations, park2017fault, saldana2017resilient}).
There, under the assumption that a bound on the total number of malicious nodes 
is known, each node eliminates neighbors taking extreme values 
according to that number at each iteration. This method is simple and suitable for 
distributed implementation. 
However, it requires the network to be relatively dense and complex.
This is partly because MSR algorithms do not have the functionality to 
detect the adversaries. 
In particular, an explicit characterization expressed by the notion of robustness of graphs has been obtained. 
Further extensions to more sparse graphs using additional trusted nodes or heterogeneous nodes have been made in, e.g., \cite{abbas2017improving, mitra2018impact, zhao2019resilient}.

On the other hand, from the security viewpoints, it is desirable to equip the
nodes with distributed algorithms for fault detection and identification (FDI).
For consensus-type problems, FDI techniques based on unknown input observers 
have been proposed in, e.g., \cite{pasqualetti2012consensus, shames2011distributed, sundaram2011distributed}. 
These schemes however impose strong assumptions that each agent should have 
the global knowledge of the topology of the entire network to afford complete detection and identification of malicious nodes
and, moreover, they must have sufficient computation resources to run a number of observers. These aspects may limit their scope of applicability in practice.

The main objective of this paper is
to address the problem of designing fully distributed
FDI methods requiring only local information by the agents.
In this setting,
each non-faulty, normal agent in the network acts as a detector,
monitoring its neighbors by iteratively exchanging
more information than in conventional consensus.
Specifically, the agents not only send their own
states but also relay their neighbors' state values.
In this way, they can
verify if the states sent by a particular neighbor
are consistent with those of others.
Such a technique is known as two-hop communication and
is commonly used in computer science
for solving the iterative approximate Byzantine consensus (IABC) problem (\cite{fischer1986easy, Lynch, sakavalas2018effects, su2017reaching}).
By introducing multi-hop communication in MSR algorithms, the authors of \cite{su2017reaching} solve the IABC problem with a weaker condition on graph structures compared to that derived under the single-hop communication model (\cite{vaidya2012iterative}).
However, the condition in \cite{su2017reaching} is computationally hard to check and is not intuitive.\footnote[1]{
	In this paper, the topology of the multi-agent system refers to the network structure representing the direct communication among the agents over, e.g., wireless channels.  On the other hand, the term \textit{effective} topology refers to the network pattern where edges represent the availability of the information of the connected agents. Hence, the effective topology of a network with two-hop communication is in general denser than the original topology. In an adversarial environment, however, the two-hop information requires additional verification of the information received. Thus, comparison of effective topologies of different algorithms can be difficult.}

For the malicious adversary model,
there are recent studies based on two-hop
communication for distributed detection in similar
problem settings, but we stress that often strong assumptions are
imposed. 
In \cite{guo2012distributed}, faulty nodes are restricted in that
they must send the true values of their neighbors,
i.e., faulty nodes cannot lie about their neighbors.
In \cite{he2013sats, zhao2014secure}, the adversarial nodes cannot be neighbors
and thus cannot collude with each other.
In \cite{zhao2018resilient}, mobile agents are introduced, randomly visiting agents so as to eventually enable them to share global information;
they carry out a certain portion of
the workload for detection.
Also related works in the area of robotics include \cite{fagiolini2009dynamic}.

We propose two schemes for distributed fault detection and resilient
consensus, which are capable to detect malicious nodes
in the network when they misbehave.
In our detection framework, every node acts as a
local monitor for their neighbors' behaviors
during the iterations of the consensus process.
It becomes crucial for each agent $i$ to
have the information about the updates of its neighbor $j$.
Such information consists of the values of agent $j$'s
neighbors, which are the two-hop neighbors of agent $i$.
In the course, we exploit the property in the malicious model that
the adversarial nodes are restricted to send the same information to its neighbors.

The key to achieve distributed detection is
to impose the network to have sufficient connectivity.
In particular, for the agents to receive trustable information
regarding its one/two-hop neighbors,
it is critical that there are multiple ways to have access
to the neighbors through the presence of
common neighbors and multiple paths.
By both schemes, we can further achieve consensus
among the non-faulty ones in a resilient manner.
We clarify tight conditions on the network
structures in terms of the graph connectivity
for the proposed algorithms.
The first scheme has a simple structure
as the normal agents can detect malicious neighbors
but must rely on a secure mobile agent to communicate the
detection information to others. 
On the other hand, the second scheme can perform fully
distributed detection.
The major difference in their requirements lies in
the necessary network structures. The first scheme
functions on networks with less connectivity than the
second one; this is because in the latter scheme, 
majority voting (\cite{parhami1994voting}) is used for agents to decide
the true values of the two-hop neighbors. We will
see however that the required connectivity level can
be less than that in MSR-based resilient consensus algorithms.
Moreover, the proposed algorithms can function properly when more than half of the nodes turn malicious under certain topologies, which is a case out of the capabilities of conventional algorithms.


The rest of this paper is organized as follows. 
In Section~2, preliminaries on graphs and the system model are introduced. Section~3 is 
devoted to the basics of the distributed detection framework with detection share. 
In Section~4, we present our main algorithm being capable of fully distributed detection of adversaries.
In both cases, we provide necessary and sufficient graph conditions for the detection schemes.
In Section~5, numerical examples are provided to illustrate 
the effectiveness of the proposed schemes. 
We conclude the paper in Section~6.
A preliminary version of this paper appeared as \cite{yuan2019}. The current paper contains all the proofs, further discussions, and simulation results.

\section{Preliminaries and Problem Setting}


\subsection{Graph Notions}
Consider the directed graph $\mathcal{G} = (\mathcal{V},\mathcal{E})$ consisting of the node set $\mathcal{V}=\{1,...,n\}$ and the edge set $\mathcal{E}\subset \mathcal{V} \times \mathcal{V}$. Here, the edge $(j,i)\in \mathcal{E}$ indicates that node $i$ can get information from node $j$. Node $j$ is said to be an in-neighbor of node $i$, and node $i$ is an out-neighbor of node $j$. The sets of in-neighbors and out-neighbors of node $i$ are denoted by $\mathcal{N}_i=\{j\in \mathcal{V}:\, (j,i)\in \mathcal{E} \}$ 
and $\mathcal{N}_i^{\textup{out}}=\{j\in \mathcal{V}:\, (i,j)\in \mathcal{E} \}$, respectively. 
The in-degree of node $i$ is given by $d_i=\left| \mathcal{N}_i\right| $. Here, $\left| \mathcal{S}\right| $ is the cardinality of a finite set $\mathcal{S}$. 
In undirected graphs, the edge $(j,i)\in \mathcal{E}$ indicates $(i,j)\in \mathcal{E}$.
A complete graph denoted by $\mathcal{K}_n= (\mathcal{V},\mathcal{E})$ is defined by $\mathcal{E} = \{(i, j)\in \mathcal{V} \times \mathcal{V} : i \neq j\}$.
A path from node $i_1$ to $i_m$ is a sequence of distinct nodes $(i_1, i_2, \dots, i_m)$, where $(i_j, i_{j+1})\in \mathcal{E} $ for $j=1, \dots, m-1$. This path is also referred to an $(m-1)$-hop path. We also say that node $i_m$ is reachable from node $i_1$.

To characterize topological properties of networks, we provide two notions here: Graph connectivity and graph robustness.
A graph $\mathcal{G}$ is said to be (strongly) connected if every node is reachable from every other node. 
A graph $\mathcal{G}$ is said to be $k$-(node) connected if it contains at least $k+1$ nodes, and does not contain a set of $k-1$ nodes whose removal disconnects the graph; and $\kappa(\mathcal{G})$ is defined as the largest $k$ such that $\mathcal{G}$ is $k$-connected. 
In particular, for the complete graph $\mathcal{K}_n$, it holds $\kappa(\mathcal{K}_n)=n-1$.
Similarly, we introduce a connectivity notion for directed graphs here. A directed graph $\mathcal{G}$ is said to have $k$-(node) connected rooted spanning trees if it contains at least $k+1$ nodes, and does not contain a set of $k-1$ nodes whose removal renders that the digraph does not have any rooted spanning tree.

In the context of resilient consensus, graph robustness
is another important notion (\cite{leblanc2013resilient}). To verify robustness of a given graph is computationally difficult since it involves combinatorial procedures.

\begin{defn} A directed graph $\mathcal{G} = (\mathcal{V},\mathcal{E})$ is said to be $(r,s)$-robust if for every pair of nonempty disjoint subsets $\mathcal{V}_1,\mathcal{V}_2\subset \mathcal{V}$, at least one of the following conditions holds:
	(i) $\mathcal{X}_{\mathcal{V}_1}^r=\mathcal{V}_1$;   (ii) $\mathcal{X}_{\mathcal{V}_2}^r=\mathcal{V}_2$;   (iii) $\left| \mathcal{X}_{\mathcal{V}_1}^r\right| +\left| \mathcal{X}_{\mathcal{V}_2}^r\right| \geq s$;
	\noindent where $\mathcal{X}_{\mathcal{V}_a}^r$ is the set of nodes in $\mathcal{V}_a$ having at least $r$ incoming edges from outside $\mathcal{V}_a$. As the special case with $s=1$, graphs which are $(r,1)$-robust are called $r$-robust.
\end{defn}

\subsection{Update Rule and Threat Model}\label{problemsetting}

Consider a time-invariant network modeled by the directed graph $\mathcal{G} = (\mathcal{V},\mathcal{E})$.
The node set $\mathcal{V}$ is partitioned into the set of normal nodes $\mathcal{N}$ and the set of adversary nodes $\mathcal{A}$. The latter set is unknown to the normal nodes at time step $k=0$. The adversary nodes in $\mathcal{A}$ try to prevent the normal nodes in $\mathcal{N}$ from reaching consensus. 
Denote by $\mathcal{A}_i[k]$ the set of indices of the adversary nodes known to or detected by node $i$ by time step $k$. This set is updated differently in the two proposed schemes and is specified later. The set of agent $i$'s neighbors not behaving adversarially is denoted by $\mathcal{M}_i[k]=\mathcal{N}_i \setminus \mathcal{A}_i[k] $.
All algorithms in this paper are synchronous. Each normal node $i$ updates its state value $x_i[k]$ by 
\begin{equation}
x_i[k+1]=\sum_{j\in \mathcal{M}_i^+[k]} \omega_{ij}[k]x_j[k],  \label{updaterule}
\end{equation}
where $\mathcal{M}_i^+[k]=\{i\}\cup \mathcal{M}_i[k]$ and $\omega_{ij}[k]=1/(1+\left| \mathcal{M}_i[k]\right| )$. 

Next, we introduce the threat model (\cite{leblanc2013resilient}) studied here.
\begin{defn}
	\textit{($f$-total / $f$-local set)}
	The set of adversary nodes $\mathcal{A}$ is said to be $f$-total
	if it contains at most $f$ nodes, i.e., $\left| \mathcal{A}\right| \leq f$.
	Similarly, it is said to be $f$-local
	if for any normal node $i\in \mathcal{N}$, it has at most $f$ adversary nodes as its in-neighbors, i.e., $\left|\mathcal{N}_i \cap \mathcal{A}\right| \leq f, \forall i \in \mathcal{N}$.
\end{defn}

\begin{defn}
	\textit{(Malicious nodes)}
	An adversary node $i\in \mathcal{A}$ is said to be a malicious node
	if it can change its own value arbitrarily,\footnote[2]{
		Here a malicious node can decide not to send any value. This corresponds to the omissive/crash model \cite{Lynch}.} but sends the same 
	value to its neighbors at each transmission.
\end{defn}

In this paper, we focus on the malicious model. This model
is reasonable in applications such as wireless sensor networks
and robotic networks, where neighbors' information is obtained
by broadcast communication or vision sensors. This class of adversaries
has not been well studied in computer science where the Byzantine model
is traditionally more common (\cite{Lynch}). 
Specifically, a Byzantine node can send different values to its different neighbors.
We now introduce the type of consensus among the normal
agents to be sought in our paper. 
Here, we use the notation $x_i^{(i)}[k]$ to indicate the state $x_i[k]$ of each normal agent $i$ stored by itself. Its formal definition is given in Section \ref{infoset}.

\begin{defn}
	If for any possible sets and behaviors of the
	malicious agents and any state values of the normal
	nodes, the following two conditions are satisfied,
	then we say that the normal agents reach 
	resilient consensus:

	1. Safety condition: All normal states remain in the interval
	$\mathcal{S}=[\min_{i\in\overline{\mathcal{V}}}x_i^{(i)}[0],\max_{i\in\overline{\mathcal{V}}}x_i^{(i)}[0]]$, where $\overline{\mathcal{V}}=\{i\in \mathcal{V}:\, x_i^{(i)}[0]\in [\overline{x}_{\min}, \overline{x}_{\max}] \}$
	determined by the initial states of all agents: $x_i^{(i)}[k] \in \mathcal{S}, \forall i \in \mathcal{N}, k \in Z_+$.

	2. Consensus condition: There exists a state $x^*\in \mathcal{S}$
	such that $\lim_{k\to \infty}x_i^{(i)}[k]=x^*,  \forall i\in \mathcal{N}$.
	
\end{defn}

Note that for our algorithms it is hard to detect adversary nodes which take extreme initial values but perform the consensus like normal nodes.
In fact, it is hard for any algorithm to detect such nodes (\cite{fagiolini2009dynamic, guo2012distributed, he2013sats, zhao2018resilient, zhao2014secure}).
Since one can never know that such a node is normal with an extreme initial value or it is simply adversarial. 
To mitigate the impact of such adversary nodes, we set the safety interval $[\overline{x}_{\min}, \overline{x}_{\max}]$ for normal nodes so that neighbors taking values outside this interval will be considered malicious.
Here, the safety condition is different from those in MSR-based works (\cite{azadmanesh2002asynchronous, dibaji2017resilient, leblanc2013resilient, mendes2015multidimensional}), where the interval $\mathcal{S}$ is set only by the initial values of the normal agents.


In this paper, we develop two distributed schemes for
achieving both detection of malicious agents and
resilient consensus under the $f$-total model.
A common characteristic of these algorithms is that they both employ two-hop communication, where each
agent transmits its own state and
the states of its neighbors. 
We will clarify the necessary network structure for
the schemes to accomplish this goal. In particular,
our schemes require less connectivity for the networks
in comparison with the conventional MSR-based algorithms.

The two proposed schemes are dealt with in Sections \ref{Secfors1}
and \ref{Secfors2}. They differ in two aspects
related to communication and network connectivity:
The first scheme is introduced more for illustrating the idea behind adversary detection via two-hop communication;
it uses a secure mobile detector to verify the detection report and send the detection information to all nodes, but the network can be more sparse.
The second scheme is the main algorithm of this paper, being capable of fully distributed
detection under more dense networks.

Before we proceed, we explain more about 
the class of MSR-based algorithms. In such algorithms,
the normal agents remove the extreme state values
at each iteration. Specifically they may remove up to $f$
largest values and $f$ smallest values. This indicates that 
the network must be at least $(2f+1)$-connected.
However, it is known that more connectivity is needed. 
In general, under the $f$-total model, resilient consensus can be
reached by MSR-based algorithms if and only if
the underlying graph is $(f +1, f + 1)$-robust (e.g., \cite{leblanc2013resilient}).

\subsection{Information Set for Two-hop Communication}\label{infoset}

In our detection frameworks, each normal node acts as a detector for its neighbors. To this end, the nodes update and exchange their information sets \textcolor{black}{
containing their own and neighbors' values, IDs and corresponding identities (normal or malicious). }
Specifically, at each time step, each node will check the information sets 
received from its neighbors by comparing them and also with the past information 
sets. Then, it determines whether the information set is from an adversary node. 
After confirming the identities of its neighbors, 
it will utilize the values of nodes that have behaved normally in
the update rule \eqref{updaterule}, \textcolor{black}{
and send a new information set containing its updated value and the values of all the neighbors from the previous time step along with corresponding identities.}

In the update rule \eqref{updaterule}, each normal agent uses its own state and its neighbors' states. However, to explicitly indicate the difference between the broadcast states and the manipulated states, we introduce two notations. 
For agent $i$, we denote its value by $x_i^{(i)}[k]$, to indicate that it is stored in agent $i$ itself and then broadcasted. 
Let $x_i^{(j)}[k]$ be the value stored by its neighbor node $j$ after receiving the broadcast value $x_i^{(i)}[k]$. If node $j$ is malicious, this information can be modified from $x_i^{(i)}[k]$ and take a different value when it is stored by agent $j$.

Now, for each normal node $i\in \mathcal{N}$, we rewrite the update scheme in \eqref{updaterule} using these notations as
\begin{equation}
x_i^{(i)}[k+1]=\sum_{j\in \mathcal{M}_i^+[k]} \omega_{ij}[k]x_j^{(i)}[k].
\label{eqn:x_update2}
\end{equation}
For each malicious node $i\in \mathcal{A}$, it can update its broadcast state arbitrarily as
\begin{equation}
x_i^{(i)}[k+1] = u_i[k],
\end{equation}
where $u_i[k]$ may even be a function of states of all nodes in the network by time step $k$. (See also Assumption \ref{cannotadd}.)

Both normal and malicious nodes transmit information sets to their neighbors. Specifically, the information set $\Phi_i[k]$ of node~$i$ to be sent to its neighbors at time $k$ contains the value of itself at time $k$ and past values of itself and its neighbors at time $k-1$ and corresponding identities of the neighbors and is set as
\textcolor{black}{
\begin{equation}
\Phi_i[k]=\left( (i,x_i^{(i)}[k|k]),\{(j,x_j^{(i)}[k-1|k])\}_{j\in \mathcal{N}_i},\mathcal{A}_i[k]\right).
\label{eqn:Phi}
\end{equation}
}
We use the notation $x_i^{(i)}[k-1|k]$ to indicate that it is in the set $\Phi_i[k]$ from time $k$. Note that $\Phi_i[k-1]$ and $\Phi_i[k]$ contain $x_i^{(i)}[k-1|k-1]$ and $x_i^{(i)}[k-1|k]$, and if node $i$ is malicious, these values may be different. 
Moreover, under the malicious model, all neighbors of agent $i$ receive 
the same information set from agent $i$.


\begin{figure}[t]
	\centering
	
	\subfigure[]{
		\includegraphics[width=1.0in]{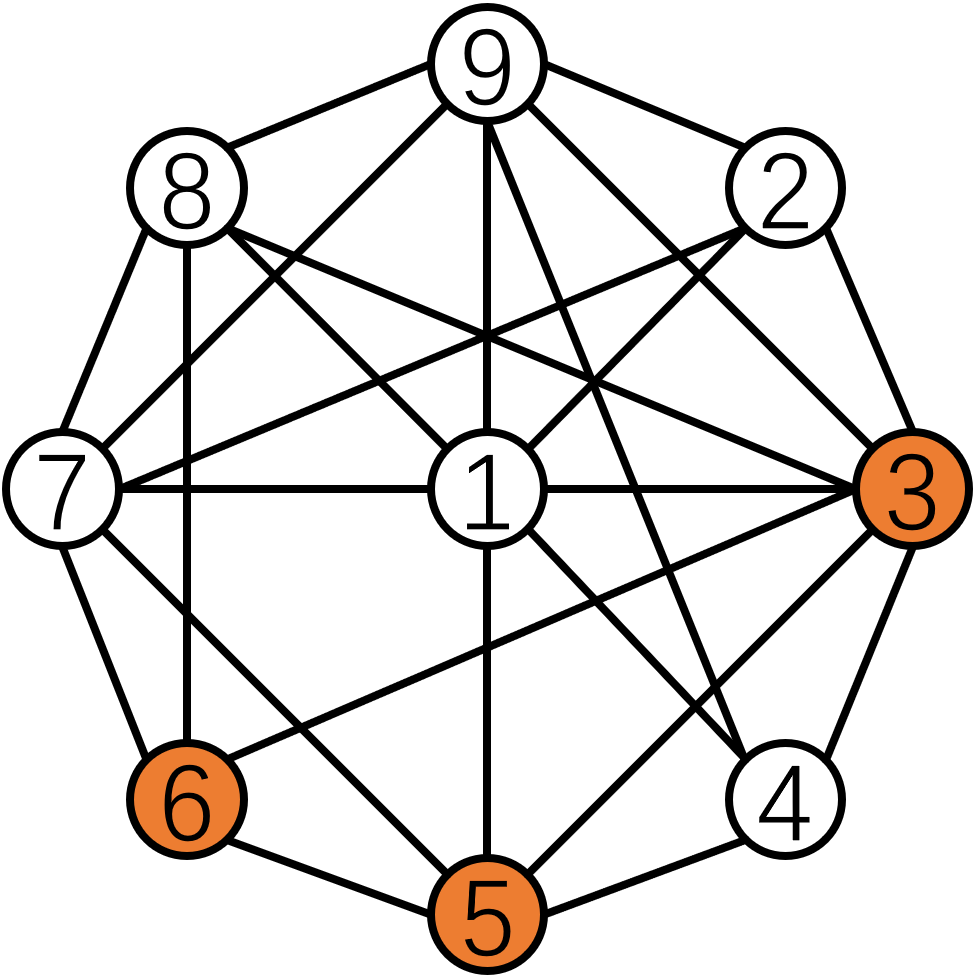}
		\label{4robust9}
	}
	\quad
	\vspace{-7pt}
	\subfigure[]{
		\includegraphics[width=1.0in]{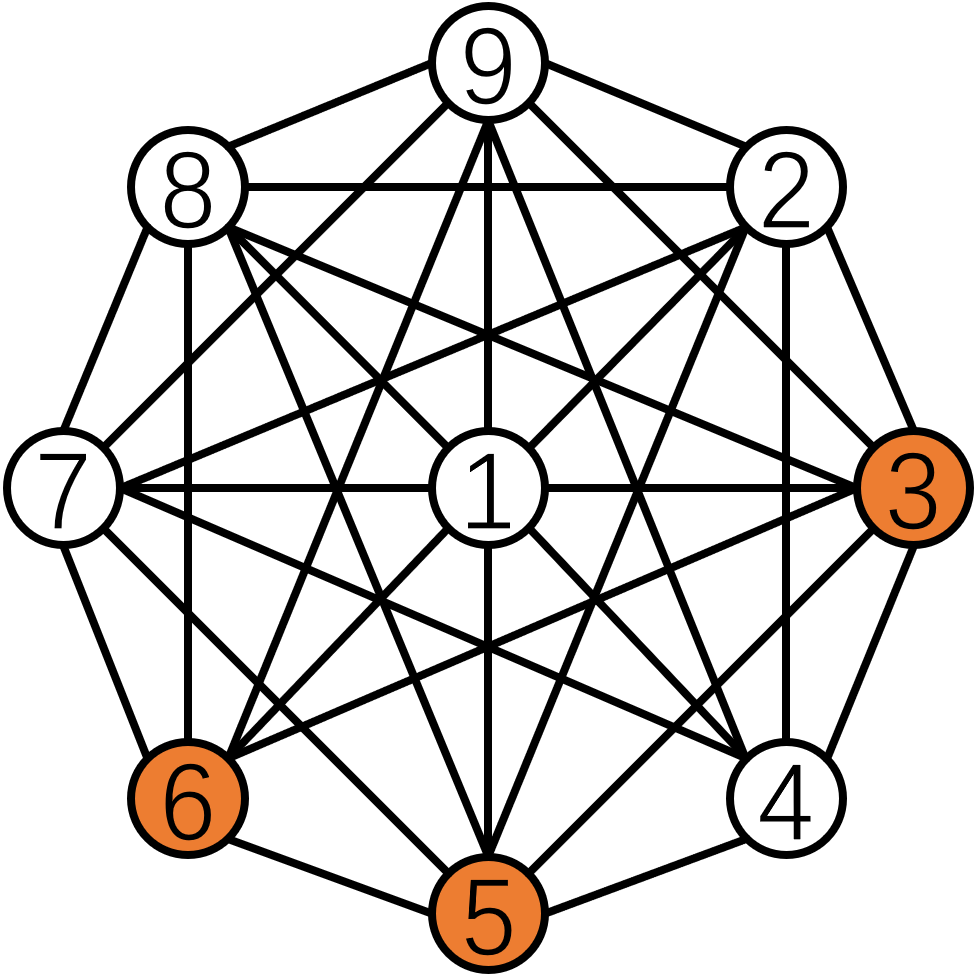}
		\label{44robust}
	}
	\vspace{-4pt}
	\caption{9-node graphs: (a) 4-connected and (b) (4,4)-robust.}
	\label{graph9}
\end{figure}

For example, consider the graph in Fig. \ref{4robust9}. Let the initial state be $x[0]=[8\ 10\ 4\ 2\ 1\ 5\ 9\ 3\ 6]^T$. The information set of node 2 at time step 1 is given by
\begin{equation}
\Phi_2[1]=\left( (2,7.4),\{(1,8), (2,10), (3,4),(7,9),(9,6)\},\mathcal{A}_2[1] \right).
\end{equation}
Here, the state of node~2 is updated by \eqref{eqn:x_update2}
as the average of the current values
with equal weights: $x_2[1] = (1/5)\sum_{j\in \mathcal{M}_2^+[0]} x_j[0]=7.4$.
This information set $\Phi_2[1]$ is sent to node~2's neighbors (i.e., nodes 1, 3, 7, and 9). After confirming that node~2 is normal through the detection algorithm, nodes 1, 3, 7, and 9 will then forward node~2's value to their neighbors at the next time step.

We introduce assumptions on the nodes' knowledge and the attacks that the malicious nodes can generate.

\begin{assum}\label{neighborinfo}
	Each normal node has access to only the information sets received from its neighbors. 
	It also has the topology information of its two-hop neighbors, i.e., the neighbors of its neighbors. 
\end{assum}

\begin{assum}\label{cannotadd}
	Each malicious node has all the information of the network (even if there is no edge from some nodes) and can manipulate its own information set in \eqref{eqn:Phi}
	before broadcasting it to the neighbors. 
	It can change the state values and IDs of its own and its neighbors.
\end{assum}

As stated in Assumption~\ref{neighborinfo}, each normal node 
has only partial knowledge about the network. 
This is actually a relaxed version of the assumption that each fault-free node knows the topology of the entire network, which is commonly made in the observer based detection works (\cite{pasqualetti2012consensus}, \cite{sundaram2011distributed}), multi-hop communication related works (\cite{Lynch}, \cite{sakavalas2018effects}, \cite{su2017reaching}), and Byzantine agreement works (\cite{tseng2015fault}). In contrast, for most of the MSR-based works (\cite{dibaji2017resilient}, \cite{leblanc2013resilient}), each fault-free node is assumed to have access to only the information from its one-hop neighbors. Thus, 
MSR-based works impose a weaker assumption than Assumption 5, however with the tradeoff of not having detection capabilities for malicious agents.
In an uncertain environment where neighbors cannot be trusted, it may not be possible to obtain accurate status regarding the two-hop neighbors. To keep the problem tractable, each node is aware of the topology up to its two-hop neighbors in this paper. This setting may be justified and of low cost as in many sensor networks, the nodes are geographically fixed and the network topology will not change.

On the other hand, a malicious node is capable 
to manipulate any value in its own information set by 
changing or deleting the value or adding some pairs of values and agent IDs.
Since the normal agents have the knowledge of the topology up to their
two-hop neighbors, attackers will be known by their direct neighbors when they do not send out their information sets,
delete values from neighbors, or add non-existing agents as neighbors.
Moreover, in the case that a malicious neighbor of some node adapts the same ID as a normal neighbor, such attacks will be detected too.

At this point, we summarize the common settings for the two schemes mentioned so far:
	(i) We deal with malicious adversary nodes (including omissive/crash model).
	(ii) The underlying network graph is time invariant.
	(iii) The update rules are synchronous.
	(iv) Both schemes are applied for scalar consensus. 

\section{Scheme 1 with Detection Share}\label{Secfors1}

In the first scheme for distributed detection and resilient consensus, the normal nodes are capable to detect malicious neighbors by using the two-hop information in undirected networks. It provides the basics for using two-hop communication in an adversarial environment, which is motivated by the works \cite{zhao2018resilient, zhao2014secure}.

\subsection{Resilient Consensus Scheme 1}
The update rule together with the detection algorithm can be outlined as follows:

\begin{scheme}
Each agent $i \in \mathcal{V}$ exchanges with its neighbors the information set $\Phi_i[k]$ in \eqref{eqn:Phi}.
Each normal agent first runs the detection algorithm in Algorithm~1. Once it detects any malicious agent in its neighbors, then the detection information is broadcasted to all agents through the secure mobile agents.
Finally, it will use the values from its normal neighbors to update its value by the update rule \eqref{updaterule}.
\end{scheme}

\begin{algorithm}[t]\footnotesize
	\caption{Detection Algorithm for Scheme 1} 
	\begin{algorithmic}
		\REQUIRE $\Phi_j[k],\Phi_j[k-1],j\in \mathcal{N}_i \cup \{i\} $
		\ENSURE IDs of the malicious neighbors
		\STATE 
		Initialization: Take the check set $\mathcal{C}_i[0]$ and malicious node set $\mathcal{A}_i[0]$ to be empty. Moreover, node $i$ receives and stores the initial states of its neighbors as $x_j^{(i)}[0]=x_j^{(j)}[0]$ for all $j\in\mathcal{N}_i$. If $x_j^{(i)}[0] \notin [\overline{x}_{\min}, \overline{x}_{\max}]$ for any $j\in\mathcal{N}_i$, node $i$ will consider node $j$ as malicious. 
		
		\STATE At each time $k$, node $i$ executes the following steps:
		
		\STATE \textcolor{black}{Let $\mathcal{A}_i[k]=\mathcal{A}_i[k-1]$. By the assumption of detection share, the malicious nodes detected at time $k-1$ is shared among all the nodes by time $k$.}
		
		\FOR{$j\in \mathcal{N}_i$}

		\IF{(Step 1) \textcolor{black}{ $\Phi_j[k]$ contains any different identities of the nodes known by node $i$ (any node in $\mathcal{A}_i[k]$ is labeled as malicious, otherwise is labeled as normal.)} } 
		\STATE it will output node $j$ as malicious. 
		\ENDIF

		\IF{(Step 2) $\Phi_j[k]$ contains IDs of neighbors of node $j$ which are different from those known to node $i$} 
		\STATE it will output node $j$ as malicious. 
		\ENDIF

		\IF{(Step 3) any value of $x_h^{(j)}[k-1|k], h \in \mathcal{N}_i$, in $\Phi_j[k]$ is not equal to the value in the check set $\mathcal{C}_i[k-1]$} 
		\STATE it will output node $j$ as malicious. 
		\ENDIF 
		
		\IF{(Step 4) $x_j^{(j)}[k|k]$ in $\Phi_j[k]$ does not follow the update rule \eqref{updaterule}} 
		\STATE it will output node $j$ as malicious. 
		\ENDIF

		\IF{(Step 5) node $j$ has not been detected as malicious by node $i$ through the 4 steps above} 
		\STATE it will output node $j$ as normal. 
		\ENDIF

		\ENDFOR
		
		\STATE Return identities (malicious or normal) of neighbors. 
		If node $i$ detects node $j$ as malicious or it receives a report on node $j$ through the detection share, it will put node $j$'s ID in $\mathcal{A}_i[k]$.
	    
		\STATE Node $i$ stores $x_j^{(j)}[k|k] $ from $ \Phi_j[k]$, $ j\in \mathcal{N}_i \cup \{i\}$, into $\mathcal{C}_i[k]$.
		
	\end{algorithmic}
\end{algorithm}

\subsection{Detection Algorithm Design}

For the first scheme, the detection share function explained below is needed for the communication among the nodes when events of detecting adversaries occur.

\begin{assum}\label{broadcast}
	Once a malicious node is detected by
	any of the normal nodes, its ID will be
	securely notified to all nodes.
\end{assum}

This type of assumptions appears in \cite{zhao2018resilient, zhao2014secure}
as well. We however stress that our
results have advantages over these works
in terms of the network structure requirements.
We will make more precise comparisons later.
In practice, for this
detection share, a certain level of resources
is necessary. \textcolor{black}{
This can be realized by introducing fault-free
mobile nodes which are appropriately distributed throughout the network and are
capable to verify if the detection reports from a node
is true or false.}
Here, we suppose that
the nodes may turn malicious over
time with the upper bound $f$ on the total
number of such nodes. Thus, the verification
must take place in real-time. Each time a node
claims to have detected a malicious neighbor,\textcolor{black}{
the mobile agent nearest to the node visits it} and verifies the evidence of the report, i.e., by collecting the information sets of the node and its neighbors of that time step.
If it finds the detection report to be valid, then it \textcolor{black}{
broadcasts the detection information to all nodes} through secure communication (\cite{lamport1982byzantine}, \cite{lindell2006composition}).
Otherwise, it broadcasts that the node sending the report is malicious.
We emphasize that these mobile agents must verify the
detection reports only when they receive from agents, 
and they need not carry out the detection of adversaries
themselves, which requires keeping track of the entire network 
all the time as in \cite{zhao2018resilient}.

We now present our distributed detection scheme in Algorithm~1. 
To ensure that all nodes follow the specified 
update rule, the normal nodes utilize the information set 
$\Phi_i[k]$ given in
\eqref{eqn:Phi} and check consistency among
the data received from their neighbors.
In Algorithm~1, step 1 is to guarantee that each normal 
node should not use the information from the nodes detected 
to be malicious by the previous time step. 
Step~2 is to prevent the malicious nodes 
from faking any neighbors.
Step~3 is to enforce the normal nodes not to
modify the values received from their neighbors. 
Finally, step~4 is to guarantee that the normal nodes 
follow the given update rule.

Our approach is distributed as
this detection scheme is implemented on each node. 
By contrast, in \cite{zhao2014secure}, a strict assumption 
on the malicious nodes is imposed so that the cooperation 
between the malicious nodes is not allowed. 
In particular, it is assumed that malicious nodes can not 
be neighbors. It is difficult to guarantee this in practice since 
clearly the identities of the malicious nodes are unknown prior to operation. Even
faults may occur simultaneously in two neighboring nodes.
In the later work (\cite{zhao2018resilient}) by the same authors, 
to relax this assumption, mobile agents are employed. Such agents
execute the fault detection and isolation (FDI) function by collecting information as
they continuously circulate within the network. 
However, for Scheme 1, the mobile agents are used only for verification of the detection reports when the events of detection occur. Later, in the next section, we will introduce another detection algorithm with the ability of fully distributed detection of adversaries.

Furthermore, a common assumption made in \cite{guo2012distributed, zhao2018resilient, zhao2014secure} is that the normal nodes form 
a connected graph. 
Although this is a necessary requirement for the normal nodes 
to achieve consensus, it is impossible to check whether 
a given graph has this property a priori even if the bound $f$ on malicious agents is known. 
In Scheme~1, we address this issue by imposing a connectivity 
condition to guarantee that the original network has a certain redundant structure.

\subsection{Necessary Graph Structure for Scheme 1}

\begin{figure}[]
	\centering
	
	\subfigure[]{
		\includegraphics[width=1.2cm]{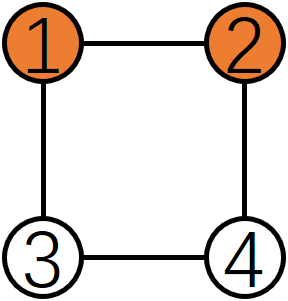}
		\label{4circle}
	}
	\quad
	\vspace{-7pt}
	\subfigure[]{
		\includegraphics[width=1.2cm]{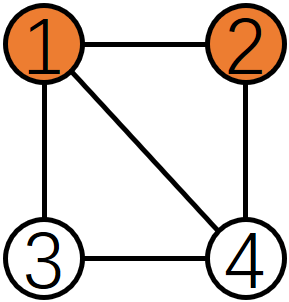}
		\label{123complete}
	}
	\vspace{-4pt}
	\caption{Example graphs: Nodes 1 and 2 are malicious. In (a), there is no common normal neighbor, while in (b), node 4 is.}
\end{figure}

Here, we introduce some conditions on the network structure
to fully utilize the detection capability of Algorithm~1. 
From the definitions of the information sets and the detection algorithm, it is clear
that a malicious node can be detected if there is at least one normal 
node among its neighbors that monitors its behavior.
However, such detection may fail if neighboring malicious 
nodes cooperate with each other. Hence,
it is critical that one or more normal nodes are present as their common neighbors.
We illustrate this point using the simple four-node network in Fig.~\ref{4circle}. 
Take nodes~1 and~2 to be malicious. 
They can cooperate as follows: Node~1 
manipulates $x_{2}^{(1)}[k-1|k]$ in its information set, 
and node~2 manipulates $x_{1}^{(2)}[k-1|k]$ in its information set. 
In this network, 
since there is no normal node having access to the information sets 
of both nodes~1 and~2, such an attack will not be detected.
Now, in the network in Fig.~\ref{123complete}, 
the normal node~4 is a common neighbor of nodes~1 and~2. 
As it has access to both $x_{1}^{(1)}[k-1|k-1]$ and $x_{2}^{(2)}[k-1|k-1]$, 
it can detect when node~1 or~2 
changes the value of the other.

The following lemma formally states this requirement and its proof can be found in Appendix A. 

\begin{lem}\label{lemma1} 
	Consider the network of nodes modeled by the undirected graph $\mathcal{G}=(\mathcal{V},\mathcal{E})$. 
	Algorithm~1 detects every pair of neighboring misbehaving nodes if and only if they have at least one normal node as their common neighbor.
\end{lem}

Since the identities of the malicious nodes are unknown, 
we must impose a connectivity requirement so that the condition 
in the lemma holds for any combination of nodes being malicious neighbors
in the network. 
The theorem below provides the main result of this section. Its proof can be found in Appendix B.

\begin{thm}\label{detect} 
	Consider the network modeled by the undirected graph $\mathcal{G}=(\mathcal{V},\mathcal{E})$ with the adversary set $\mathcal{A}$ to be an $f$-total malicious set. Suppose that Assumptions \ref{neighborinfo}, \ref{cannotadd} and \ref{broadcast} hold. Then, under Scheme~1, the following hold.

	(a) All malicious nodes that behave against the given update rule \eqref{updaterule} are detected if and only if 
	for every pair of neighboring nodes, they have at least $f-1$ two-hop paths connecting them.
	
	(b) Under the condition of (a), normal nodes can achieve resilient consensus if $\mathcal{G}$ is ($f+1$)-connected.
\end{thm}

We note that in the case of undirected graphs, for a pair of neighboring nodes to share common neighbors is equivalent to having two-hop paths connecting them. For directed graphs, we need to be more careful as we will see in Section \ref{Secfors2}.
Note that the conditions in Theorem \ref{detect} do not require dense graph structures.
For example, we can check by inspection that both graphs 
in Fig.~\ref{graph9} satisfy the conditions for the case with $f=3$. 


Our study is motivated by the MSR algorithms studied in, 
e.g., \cite{dibaji2018resilient, leblanc2013resilient}. 
There, the notion of graph robustness is shown to be critical 
to guarantee consensus in the presence of malicious agents. This can be achieved by the nodes removing extreme values of neighbors while no detection is performed.
On the other hand, our approach is to achieve adversary detection for resilient consensus, and this is realized by using extended information sets. As a result, the connectivity requirement becomes less restrictive compared to the MSR algorithms though the necessary resources for communication and computation are higher.

As mentioned before, for the MSR algorithms under 
the $f$-total malicious model, resilient consensus
is guaranteed if and only if the graph is $(f+1,f+1)$-robust.
It is known from \cite{leblanc2013resilient} that such a 
graph has the property of being $(f+1)$-connected, but the converse does not hold in general.
The difference between these classes of graphs can be checked
by the two graphs in Fig.~\ref{graph9}.
The graph in (a) is 4-connected and moreover satisfies the two-hop condition in Theorem \ref{detect} for $f=3$.
The one in (b) on the other hand is (4,4)-robust as required by the MSR algorithm with $f=3$.
Hence, when the number of malicious nodes 
in the network is the same, the constraint
on the graph structure for Scheme~1 is less stringent
than that for MSR algorithms.

\section{Scheme 2 with Fully Distributed Detection}\label{Secfors2}

Having the basics established for the detection via two-hop communication, we present our main result on Scheme~2 for resilient consensus in this part. Compared to Scheme~1, it is notable that the new scheme can achieve fully distributed detection for each normal node without the use of detection share.
This feature can be realized by introducing majority voting (\cite{blahut1983theory, parhami1994voting}) and requiring a more dense graph structure. 
Then we provide a necessary and sufficient condition on graph structures for the detection part of Scheme~2. 
Lastly, we show that Scheme~2 can tolerate more malicious nodes in both complete networks and incomplete networks compared to MSR-based algorithms.

\subsection{Resilient Consensus Scheme 2}
The update rule together with the detection algorithm can be outlined as follows:

\begin{scheme}
Each agent $i \in \mathcal{V}$ sends its out-neighbors the information set $\Phi_i[k]$ in \eqref{eqn:Phi}.
After receiving the information sets from its in-neighbors, it first runs the detection algorithm in Algorithm~2. If it detects any malicious neighbors, it reports this detection information to its out-neighbors. 
Then, it will utilize the values from its normal neighbors to update its value by the update rule \eqref{updaterule}.
\end{scheme}


It is important to note that malicious nodes may send fake detection reports to normal nodes in this scheme, which is different from Scheme~1 where the secure detection share is utilized. Hence, the normal nodes must verify if each received data is authentic.

\subsection{Detection Algorithm Design}

Malicious neighbor sets here enable the normal nodes to keep track of its neighbors identified to be malicious. 

\begin{defn}
	\textit{(Malicious neighbor set $\mathcal{A}_i[k]\subset \mathcal{V}$)} Once node $i$ detects any malicious node among its neighbors at time $k$, it puts the node's ID in $\mathcal{A}_i[k]$. Also, once node $i$ receives at least $f+1$ detection reports on some node, it will put the node's ID in $\mathcal{A}_i[k]$. This set is accessible only to node $i$ itself.
\end{defn}

From Scheme~1, we see that the information set plays a crucial role in our detection framework. 
For any normal node $i$ to verify the identity of its neighbors, the
information sets of the neighbors need to be investigated in two parts: (i) the current value, i.e., if it is updated according to the given update rule; (ii) the past values, i.e., if they are manipulated and different from the true values of the corresponding nodes.

The two proposed schemes share common features in checking the neighbors' current values. 
Nevertheless, certain difference lies between the two schemes about checking the past values of the neighbors.
With the detection share function in Scheme~1, normal node $i$ needs to check whether its own past value is manipulated in the information sets of its one-hop neighbors. 
By contrast, in Scheme~2, node $i$ should check whether any entries of the past values are manipulated in the information sets of its one-hop neighbors. 
Thus node $i$ needs to obtain the true state values of its neighbors' neighbors, i.e., two-hop neighbors. 
For this purpose, the information sets of its one-hop neighbors must be utilized. Among the multiple information sets containing the state value of its two-hop neighbor $h$, there may be some malicious node relaying a wrong value of node $h$. Thus node $i$ needs to carry out a majority voting on the true state value of two-hop neighbor $h$ through all the one-hop neighbors' information sets containing the value of node $h$. Here, majority voting means that if node $i$ receives $m$ values of node $h$, among the $m$ values, if more than $m/2$ values are the same, then node $i$ will take this value as the true value of node $h$. 

\textcolor{black}{
To fully utilize the capability of Algorithm~2, we here define the necessary graph condition for Algorithm~2.
}

\begin{defn} \textcolor{black}{
	A directed graph $\mathcal{G} = (\mathcal{V},\mathcal{E})$ is said to satisfy the condition for Algorithm~2 with parameter $f$ if for any in-neighbor $j\in \mathcal{N}_i$ of node $i$, and any $h\in \mathcal{N}_j,h\neq i$, one of the following conditions holds:}
	\begin{enumerate}
		\item $h\in \mathcal{N}_i$;
		\item $h\notin \mathcal{N}_i$, and there are at least $2f+1$ directed two-hop paths from $h$ to $i$ (including the one through $j$).
		
	\end{enumerate}
\end{defn}

We now present our distributed detection in Algorithm~2. 
Here, each node $i$ performs majority voting on two things: the nodes' values and detection information. Since we consider the $f$-total/$f$-local model in this paper, at most $f$ values could be false in the neighborhood of node $i$. Thus if node $i$ receives the same information from at least $f+1$ distinct neighbors, it considers this information trustable. After obtaining the true values of its one-hop neighbors and two-hop neighbors, it follows the same detection procedures as the ones in Scheme~1.

\begin{algorithm}[t]\footnotesize
	\caption{Detection Algorithm for Scheme 2} 
	\begin{algorithmic}
		\REQUIRE $\Phi_j[k],\Phi_j[k-1],j\in \mathcal{N}_i\cup \{i\} $
		\ENSURE IDs of the malicious neighbors 
		
		\STATE 
			Initialization: Node $i$ creates a check set $\mathcal{C}_i[k-1]$ to store the values of one-hop neighbors $x_j^{(j)}[k-1]$ and also the values of two-hop neighbors $x_h^{(h)}[k-1]$ at time $k-1$. Take the check set $\mathcal{C}_i[0]$ and malicious neighbor set $\mathcal{A}_i[0]$ to be empty. Moreover, node $i$ receives and stores the initial states of its neighbors as $x_j^{(i)}[0]=x_j^{(j)}[0]$ for all $j\in\mathcal{N}_i$. If $x_i^{(j)}[0] \notin [\overline{x}_{\min}, \overline{x}_{\max}]$ for any $j\in\mathcal{N}_i^{\textup{in}}$, node $i$ will consider node $j$ as malicious.

		\STATE At each time $k$, node $i$ executes the following steps:
		
		\STATE For each value of two-hop neighbor $x_h^{(h)}[k-1]$, it gathers $x_h^{(j)}[k-1|k]$ from $ \Phi_j[k]$, $j\in \mathcal{N}_i$, does the majority voting on the value of $x_h^{(h)}[k-1]$ and stores $x_h^{(h)}[k-1]$ into  $\mathcal{C}_i[k-1]$.
		
		\STATE \textcolor{black}{Let $\mathcal{A}_i[k]=\mathcal{A}_i[k-1]$. Once node $i$ receives at least $f+1$ detection reports on some node, it puts the node's ID in the malicious neighbor set $\mathcal{A}_i[k]$. }
		

		\FOR{$j\in \mathcal{N}_i$} 
		
		\STATE \textcolor{black}{  \textbf{Steps 1-5 in Algorithm 1.} }
		%
		%
		%
		%
		%
		%
		%
		%
		
		%
		%

		\ENDFOR
		
		\STATE Return identities (malicious or normal) of neighbors. 
		\STATE Node $i$ stores $x_j^{(j)}[k|k] $ from $ \Phi_j[k], j\in \mathcal{N}_i \cup \{i\} $, into $\mathcal{C}_i[k]$.
		
	\end{algorithmic}
\end{algorithm}

\subsection{Necessary Graph Structure for Scheme 2}

In directed networks, the necessary condition for node $i$ to detect its neighbor $j$ when it misbehaves is the following:
node $i$ has full access to the information that node $j$ must use to update its value if node $j$ is normal, that is,
the true values $x_h^{(h)}[k-1]$, $\forall h\in \mathcal{N}_j$, used in the control input of node $j$ \textcolor{black}{ 
and the correct detection information of the two-hop neighbor $h$.	
}

Similar to Scheme~1, we must impose a connectivity requirement on every node and its neighbors for Scheme~2, such that the detection is guaranteed for any combination of nodes being malicious in the network. The following theorem is the main result of this section.

\begin{thm}\label{detect2} 
	Consider the network modeled by the directed graph $\mathcal{G}=(\mathcal{V},\mathcal{E})$ where the adversary set $\mathcal{A}$ follows the $f$-total malicious model. Suppose that Assumptions \ref{neighborinfo} and \ref{cannotadd} hold. Then, under Scheme~2, the following hold.

	(a) All malicious nodes that behave against the given update rule \eqref{updaterule} are detected if and only if $\mathcal{G}$ satisfies the condition for Algorithm~2 with parameter $f$.
	
	(b) Under the condition of (a), normal nodes can achieve resilient consensus if $\mathcal{G}$ has $(f+1)$-connected rooted spanning trees.
\end{thm}

\begin{proof} \emph{(a) Necessity:} 
	We prove by contradiction. First, suppose that there is a node $h$ in $ \mathcal{N}_j$ but not in $\mathcal{N}_i$, and there is no two-hop path from node $h$ to node $i$. Take node $j$ to be malicious. In this case, it can change $x_h^{(j)}[k-1]$ in $\Phi_j[k]$ arbitrarily, but such attacks cannot be detected by node $i$ obviously since node $i$ cannot obtain any information from node $h$ through the information sets that it receives.
	
	Next, suppose that there is a node $h\in \mathcal{N}_j$ with $h\notin \mathcal{N}_i$, and that there are at most $2f$ two-hop paths from node $h$ to node $i$ including the path containing node $j$. Take node $j$ to be malicious. In this case, node $i$ will get copies of $x_h^{(h)}[k-1]$ from at most $2f$ different information sets. Here, note that node $i$ can also obtain copy of $x_h^{(h)}[k-1]$ from $\Phi_j[k]$, i.e., $x_h^{(j)}[k-1]$. Among the $2f$ copies of $x_h^{(h)}[k-1]$, no majority is guaranteed when we consider the worst case. That is, there may be $f$ identical values created by malicious nodes and $f$ identical values created by normal nodes. Thus node $i$ cannot determine which one is actually the true value of $x_h^{(h)}[k-1]$. Hence, node $j$ cannot be detected by node $i$. 
	
	\emph{Sufficiency:} 
	For detection, we prove sufficiency by showing that node $i$ can confirm the true value of every entry of the information set $\Phi_j[k]$ of neighbor node $j$ by obtaining the true value $x_h^{(h)}[k-1]$ of every neighbor $h\in \mathcal{N}_j$ of node $j$, from the previous time step $k-1$. \textcolor{black}{ Moreover, node $i$ can obtain the correct detection information of its two-hop neighbors at time $k+1$. Then we can prove that node $i$ will detect node $j$ at time $k+1$ if node $j$ sends out faulty $\Phi_j[k]$.}
	
	For node $i$, consider the following two cases separately: (i) only condition~1 holds; (ii) only condition~2 holds.
	
	(i) In the case where $h\in \mathcal{N}_i$, it is clear that node $i$ can receive the true value of $x_h^{(h)}[k-1]$ from $\Phi_h[k-1]$ and have the right detection information of its one-hop neighbor $h$ at time $k+1$. 
	
	(ii) Suppose that $h\notin \mathcal{N}_i$, and there are at least $2f+1$ directed two-hop paths from node $h$ to node $i$. In this case, there is some normal node $l\in \mathcal{N}_h^{\textup{out}} \cap \mathcal{N}_i$ which carries the true value of $x_h^{(h)}[k-1]$ in its information set at time $k$. 
	If the majority of the $2f+1$ paths from $h$ to $i$ contains nodes as $l$, then node $i$ can get the true value of $x_h^{(h)}[k-1]$. 
	Since there are at most $f$ malicious nodes among the in-neighbors of node $i$, and there are at least $2f+1$ directed two-hop paths from $h$ to $i$ including the path containing node $j$, we have the needed majority.
	
	\textcolor{black}{
	We can apply the same analysis on the detection information of node $i$'s two-hop neighbors. In the same case, if node $h$ sends out faulty $\Phi_h[k-1]$, then it is detected by its one-hop neighbors at time $k$. Recall that there are at least $2f+1$ directed two-hop paths from node $h$ to node $i$ and at most $f$ malicious nodes among the in-neighbors of node $i$. Thus, node $i$ can obtain the correct identities of its two-hop neighbors by majority voting at the receiving stage of time $k+1$. }
	
	\textcolor{black}{
	Therefore, node $i$ knows the true value of $x_h^{(h)}[k-1]$ and obtains the correct detection information of its two-hop neighbors $h$ before running the detection algorithm at time $k+1$. Thus if node $j\in \mathcal{N}_i$ sends out faulty $\Phi_j[k]$ by possible manipulation including modifying the entry of $x_h^{(j)}[k-1]$ in $\Phi_j[k]$, by simply breaking the update rule, or by sending false information on the identity of node $h$, then node $i$ will detect. }
	
	\emph{(b)} Malicious nodes will be detected immediately once they misbehave. Thus misbehaviors of malicious nodes cannot affect normal nodes since normal nodes exclude values from detected malicious nodes. Hence, the safety condition is guaranteed. Moreover, by the graph $\mathcal{G}$ having $(f+1)$-connected rooted spanning trees, after removing $f$ malicious nodes, the subgraph of normal nodes contains at least one rooted spanning tree. Therefore, resilient consensus is achieved.
\end{proof}

\begin{rem}\label{flocal}
	We must highlight that this result for the $f$-total model
	can be easily extended to the case of $f$-local model, which is more adversarial because more than $f$ malicious agents in total may be in the entire network.
	Actually, the necessary and sufficient condition
	for Scheme~2 to detect malicious nodes that behave against
	the given update rule \eqref{updaterule} under $f$-local model is the
	same as the condition stated in Theorem \ref{detect2}.
	The flow of proof follows along similar lines as the one shown above.
\end{rem}

\subsection{Discussion}\label{discuss2}

\begin{figure}[t]
	\centering
	
	\subfigure[]{
		\includegraphics[width=1.8cm]{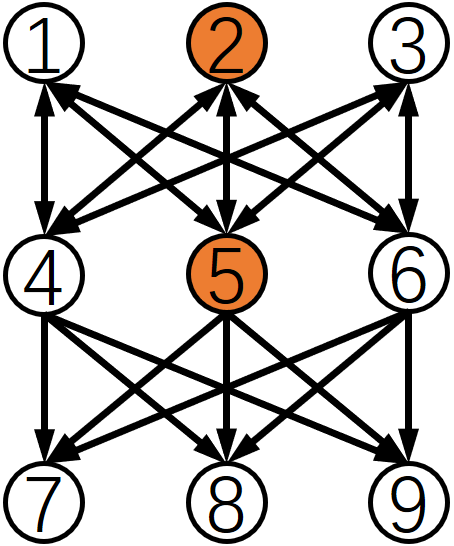}
		\label{9nodesforp2}
	}
	\quad
	\vspace{-7pt}
	\subfigure[]{
		\includegraphics[width=2.5cm]{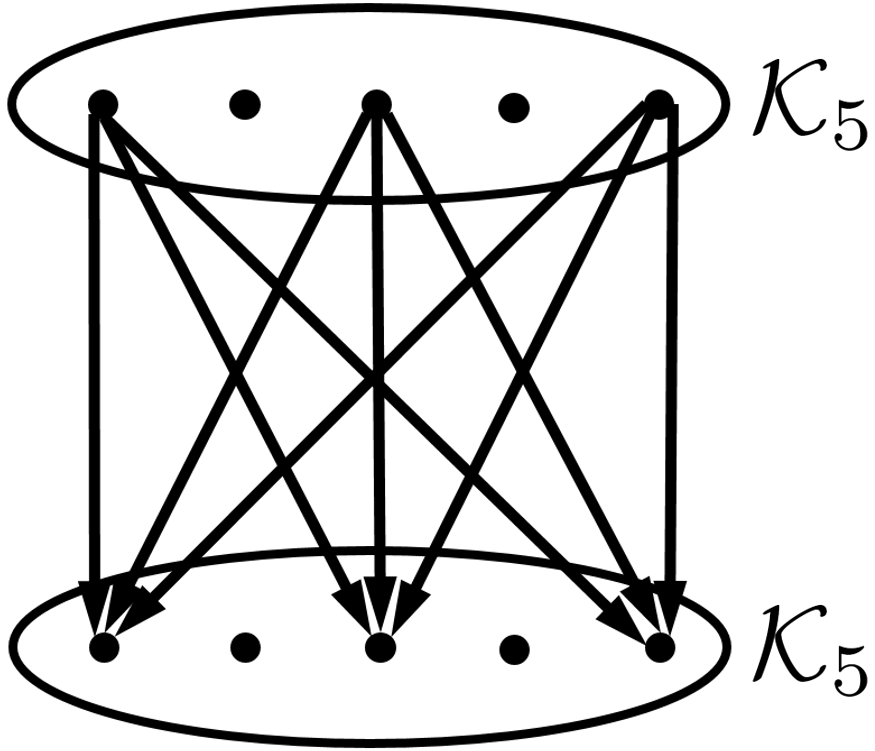}
		\label{connectcomplete}
	}
	\vspace{-4pt}
	\caption{Example graphs satisfying the criteria for Scheme~2.}
	
\end{figure}

The conditions for Scheme~2 guarantee that for any out-neighbor $i$ of node $j$, it has full access to the broadcast values of in-neighbors of node $j$. It gains the values either by being a direct neighbor of node $j$'s in-neighbors or through majority voting over at least $2f+1$ paths. For the latter case, the majority of the voting always exists and is correct since there are at most $f$ malicious nodes in the neighbors of node $i$ due to the $f$-total/$f$-local model. Also, the majority voting enables node $i$ to verify if a detection report on its two-hop neighbors is valid.
In the computer science literature, similar redundancy schemes are often used to provide security and reliability to systems. For example, if a transmission system is designed to tolerate up to $f$ failures, it must have $2f+1$ copies of the transmitted information along with majority voting for verification (\cite{blahut1983theory}).

Here we provide some example graphs satisfying the conditions in Theorem \ref{detect2}.
The network in Fig. \ref{9nodesforp2} satisfies the conditions for Scheme~2 under the $1$-local model, i.e., there is at most one malicious node in the neighbors of each normal node. 
Moreover, there is a characteristic three-layer structure. We can extend this idea to the cases with any $f$. Each layer should have $2f+1$ nodes for $f$-total/$f$-local model. Each node in one layer should be connected with every node in the neighbor layers and have no connection with the nodes in its own layer. This structure can also have many layers as long as the $f$-total or $f$-local set is satisfied for each $i\in \mathcal{N}$.
Furthermore, combining this structure with complete subgraphs (i.e., cliques), we can have graph structures like Fig.~\ref{connectcomplete}, which satisfies the conditions for Scheme~2 as well.

It is observed that each node in a complete graph can detect every malicious node since it has access to the state value of every node in the network. Thus we can enhance the performance of Scheme~2 by introducing nodes having such properties. Node $i$ is said to be a full access node if it is an out-neighbor of all other nodes in the network, i.e., $d_i=n-1$.

It is important to note that we do not assume such full access nodes to be normal. As long as the conditions for Scheme 2 are met, a full access node can also be detected by its normal neighbors when it behaves maliciously. This setting is different from the authorized central nodes and the mobile detectors in \cite{zhao2018resilient}, which randomly visit each node in the network and are assumed to be fault-free. Nevertheless, a full access node has the following property when it is normal. This result can be easily proved by Theorem \ref{detect2} since any node in the network is an in-neighbor of the full access node.

\begin{cor}\label{fullaccessnode}
	A normal full access node can detect any node that behaves against the update rule \eqref{updaterule} in the network under Scheme~2.
\end{cor}

As a result, Scheme~2 can guarantee resilient consensus in incomplete networks when the majority of the nodes are normal, if a full access node is deployed properly in the network. For example, the five-node network in Fig.~\ref{5nodesforp2} could tolerate two malicious nodes when the conditions for $1$-local are met except for the full access node 1. In the same graph, if only node 1 becomes malicious and the conditions for $1$-local are also met for other nodes, then resilient consensus is still guaranteed.

The use of full access nodes may be difficult in practice.
In the simulation, we will examine another approach
to enhance the connectivity of the network by the introduction
of relay nodes. Such nodes are limited in number, but forward
the received messages with stronger transmission power so that
the messages reach more nodes in the system.
The use of relay nodes has been well studied for wireless sensor
networks (e.g., \cite{zhang2007fault}).

\begin{figure}[]
	\begin{center}
		
		\includegraphics[width = 1.5cm ]{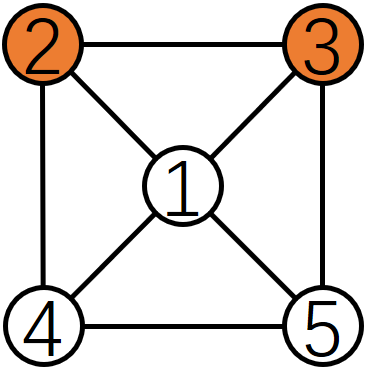}
		
		\vspace{-5pt}
		\caption{Undirected graph of 5 nodes with nodes 2 and 3 to be malicious.}\label{5nodesforp2}
		
	\end{center}
\end{figure}

Now, by the role that full access nodes can play, we can derive the maximum tolerable number of malicious nodes in incomplete graphs for Scheme~2. See Appendix C 
for the proof.

\begin{prop}\label{minimumgraph}
	Consider the network of nodes modeled by the incomplete directed graph $\mathcal{G}=(\mathcal{V},\mathcal{E})$ with $n$ nodes. Assume that it meets the conditions for Scheme~2 given in Theorem \ref{detect2} (under the $f$-total model). Then Scheme~2 detects all the $f$ malicious nodes in the network and guarantees resilient consensus only if $n>2f$.
\end{prop}

As we have seen in Proposition \ref{minimumgraph}, Scheme~2 can tolerate $n\geq 2f+1$ in incomplete networks if full access nodes are deployed properly. For Scheme~1, we can similarly obtain a necessary bound on the number of malicious nodes as $n\geq f+3$ by following the proof technique of Proposition \ref{minimumgraph}. These bounds are for incomplete graphs and are conservative if applied to complete graphs.
For Scheme 2, the bound for complete graphs is shown to be $n\geq f+2$ as stated in the following corollary of Theorem \ref{detect2}. This bound in fact holds for Scheme~1 too, which can be derived as a corollary of Theorem \ref{detect}.

\begin{cor}\label{completegraph2}
	For a complete graph $\mathcal{K}_n$, it can tolerate $f\leq n-2$ malicious nodes in the graph for the normal ones to reach resilient consensus by using Scheme~2.
\end{cor}

We summarize the maximum tolerable numbers of malicious nodes in complete graphs for several algorithms in Table \ref{table1}. 
In computer science (\cite{bonomi2019approximate, leblanc2013resilient, Lynch}), a common limitation is that MSR-based algorithms have the maximum tolerable number of malicious nodes  $n\geq 2f+1$ only for complete graphs. 
In comparison, it is clear that the proposed Schemes~1 and 2 can tolerate more adversaries.

We conduct some comparisons between Schemes 1 and 2. Their differences are shown in Table \ref{table2}.
Observe that Scheme 1 requires less connections in graphs compared to Scheme 2. For example, in the 9-node networks in Fig. \ref{graph9fors2}, networks (a) and (b) satisfy the conditions for Schemes 1 and 2 under 2-total malicious model, respectively.
It is observed that in this case strictly more connections are required for Scheme~2.
Thus, in general networks without any full access node, the requirement for Scheme~1 is easier to meet than that for Scheme~2.

\begin{table}[t]\scriptsize
	\centering
	\caption{Tolerable number of malicious nodes for complete graphs.}
	
	\vspace{0.15cm}
	\setlength{\tabcolsep}{2.8mm}{
		\begin{tabular}{cccc} 
			\toprule
			
			Scheme 1& Scheme 2& W-MSR  &MSR \\[0.5ex]  
			\midrule
			
			$n\geq f+2$ & $n\geq f+2$ & $n\geq 2f+1$ & $n\geq 2f+1$ \\

			\bottomrule
	\end{tabular}}
	\label{table1}
	
	\vspace{-2pt}
\end{table}

\begin{table}[t]\scriptsize
	\begin{center}
		\caption{Differences between Schemes 1 and 2.}
		
		\vspace{0.15cm}
		\setlength{\tabcolsep}{2mm}{
			
			\begin{tabular}{ccc} 
				\toprule
				
				& Scheme 1& Scheme 2\\[0.5ex]  
				\midrule
				
				Network& Undirected& Directed\\
				\midrule

				\tabincell{c}{Distributed\\detection} & \tabincell{c}{With \\detection share}& \tabincell{c}{Fully distributed \\detection}\\
				
				\midrule
				
				\tabincell{c}{Malicious\\model} & $f$-total& \tabincell{c}{$f$-total\\$f$-local}\\

				\midrule
				
				\tabincell{c}{Graph\\condition} & \tabincell{c}{\includegraphics[width = 1.5cm]{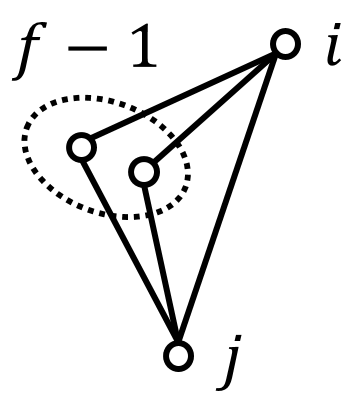}} & \tabincell{c}{\includegraphics[width = 2.6cm]{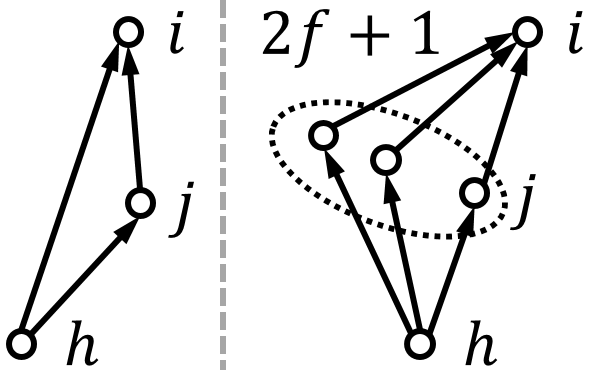}} \\
				
				\tabincell{c}{(detection)} & (a)& (b)\\
				
				\bottomrule
		\end{tabular}}
		\label{table2}
		
	\end{center}
	
	\vspace{-3pt}
	
\end{table}

We remark that for Scheme 2 under $f$-local model the tolerable number of malicious nodes in incomplete graphs could be more than the bound $n\geq 2f+1$. As mentioned in Remark \ref{flocal} and Corollary \ref{fullaccessnode}, full access nodes can detect any malicious node in the network. Thus in dense graphs which are close to complete graphs, Scheme~2 can function properly even when more than half of the nodes turn malicious. For example, in the 9-node incomplete network in Fig.~\ref{9nodes4malicious}, Scheme~2 performs well even when there are 6 malicious nodes. More details are discussed in the numerical examples.

\section{Numerical Examples}

In this section, we demonstrate the performance of
the proposed detection schemes through numerical examples.
We first use small-scale networks to verify the theoretical
results and then conduct extensive simulations based on
geometric random graphs of larger scale.

\subsection{Resilient Consensus under Scheme 1}

\begin{figure}[t]
	\centering
	
	\subfigure[]{
		\includegraphics[width=1.0in]{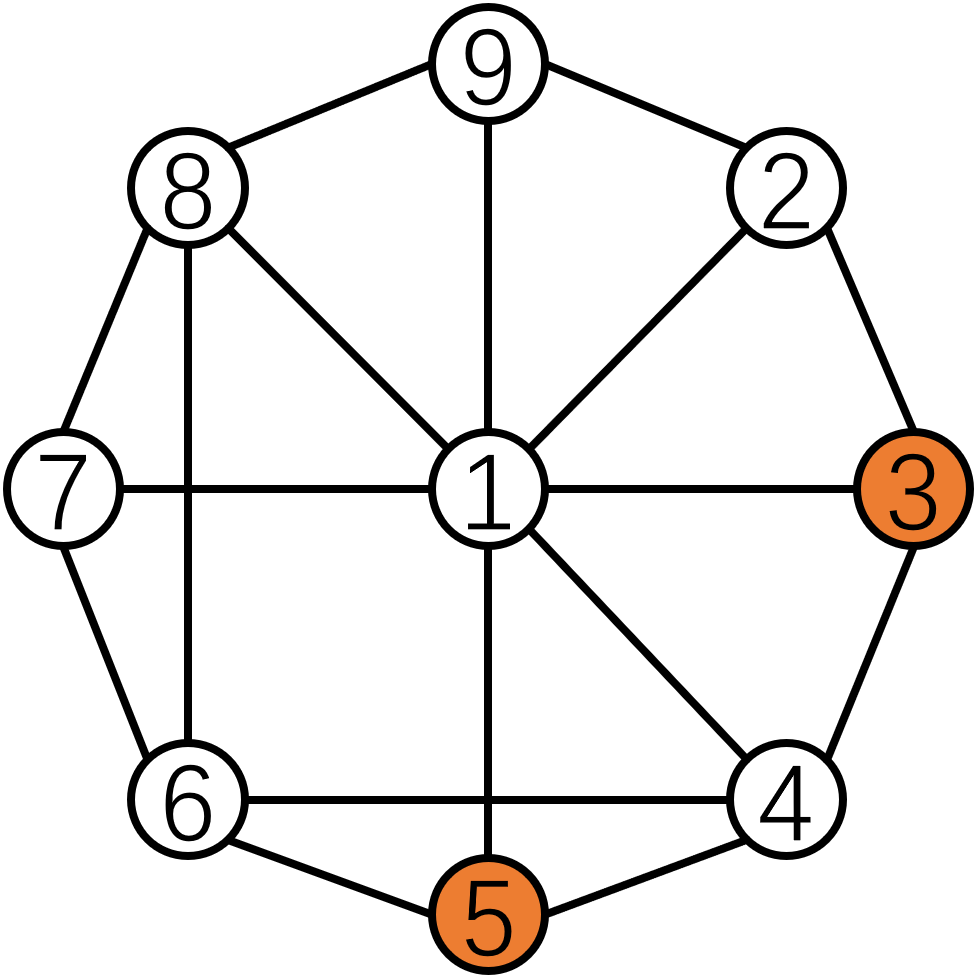}
		\label{2total}
	}
	\quad
	\vspace{-7pt}
	\subfigure[]{
		\includegraphics[width=1.0in]{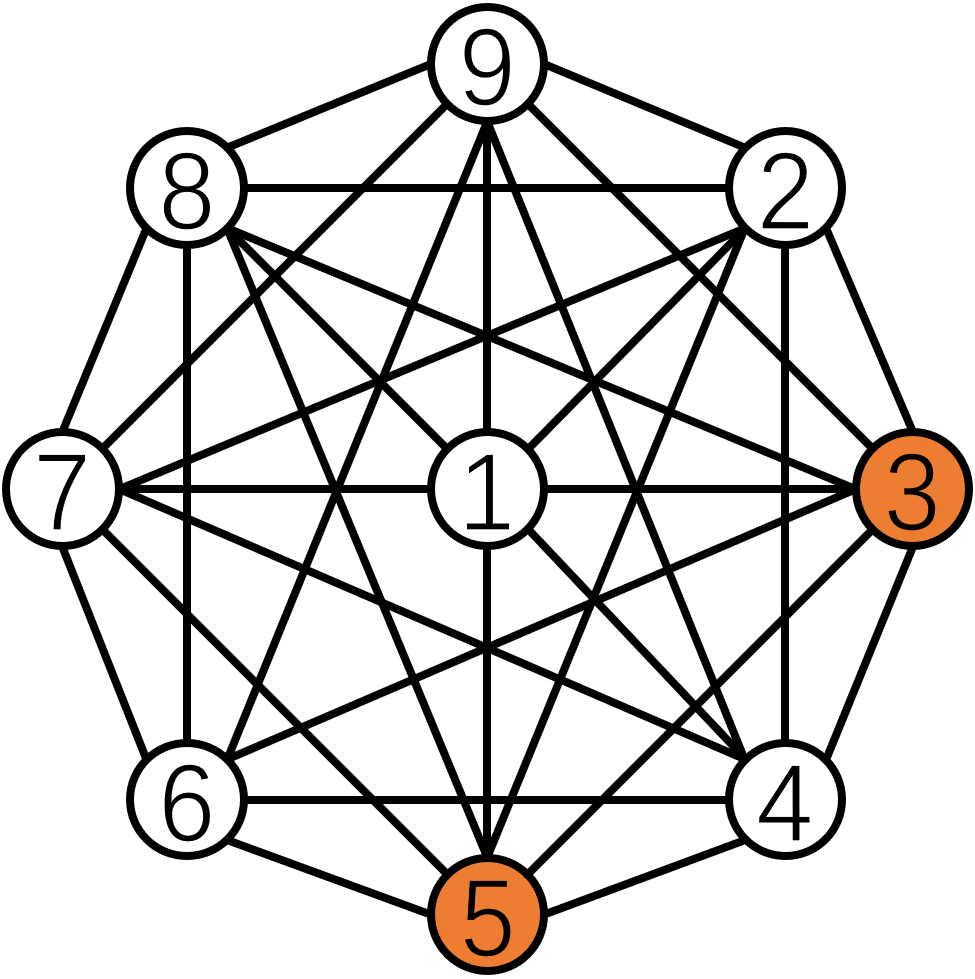}
		\label{2local}
	}
	\vspace{-4pt}
	\caption{9-node graphs under 2-total malicious model.}
	\label{graph9fors2}
\end{figure}

Consider the network shown in Fig.~\ref{4robust9}.
It is a 4-connected graph with at least two two-hop paths connecting every pair of neighbors.
Given these properties, Theorem~\ref{detect} indicates that Scheme~1 can detect and remove at most three malicious nodes, i.e., $f\leq 3$, and resilient consensus is guaranteed. 
Here, we set nodes~3, 5, and 6 to be malicious as indicated in orange in Fig.~\ref{4robust9}.

First, we examined the case without attacks. The time responses
of the states of all nodes arriving at consensus are shown in Fig.~\ref{fig4}. 
Next, in attack scenario 1, nodes~3 and~6 try to cooperate 
to avoid being detected. The simulation result is shown 
in Fig.~\ref{fig5}. 
By time $k=4$, consensus is almost achieved 
among normal nodes, but the malicious nodes start
to manipulate their information sets.
Specifically, node 3 changes the past value received from node 6 and similarly node 6 changes the past value received from node 3.
These attacks are quickly detected.  
We indicate the events of malicious node detections by dashed lines. Here, for instance, the dashed line at time 5 indicates that nodes 3 and 6 are detected as malicious.
In this case, we observe that Scheme 1 performs well. 
Finally, we note that for the W-MSR algorithm from
\cite{leblanc2013resilient}, the 4-connected network
in Fig.~\ref{4robust9} is not sufficient to achieve resilient consensus.


\subsection{Resilient Consensus under Scheme 2}

Next, we consider the network shown in Fig.~\ref{9nodes4malicious} with 5 in-coming edges of node~1 removed from the complete graph $\mathcal{K}_9$.
It satisfies the condition for Scheme~2 under $1$-local malicious model in Theorem \ref{detect2} for non-full access node~1. Here, we set nodes~2, 3, 4, 5, 6 and 7 to be malicious and the initial state to be $x[0]=[8\ 10\ 4\ 2\ 1\ 5\ 9\ 3\ 6]^T$.

First, we examined the case without attacks. The time responses
of the states of all nodes are shown in Fig.~\ref{fig14}. 
Next, in attack scenario 1, malicious nodes~2, 3, 5, 6 and 7 manipulate 
their own values; also, node~4 is malicious and keeps using the values received from them. The simulation result is shown 
in Fig.~\ref{fig15}. 
At time $k=3$, these attacks start, but are immediately detected at the next time step, with normal nodes not affected.
In both cases, the normal nodes achieve consensus.

Now, we discuss the applicability of the MSR algorithm 
from \cite{leblanc2013resilient} under
the same network. As discussed earlier,
for this algorithm, the connectivity structure of the network in Fig.~\ref{9nodes4malicious} 
is not sufficient for tolerating 6 malicious nodes. In fact, for a network with 9 nodes, even if it is a complete graph, only up to 4 malicious nodes can be tolerated (\cite{leblanc2013resilient}). In general, it is impossible for MSR algorithms to function properly when more than half of the nodes are malicious.
Moreover, we also analyze how the iterative approximate Byzantine consensus (IABC) algorithm from \cite{su2017reaching} performs under the same network. Consider the case when each node knows the topology of two-hop neighbors and the relay depth is two-hop, as we assume for Scheme~2 here. 
To meet the condition in \cite{su2017reaching}, node~8 should be connected with at least 9 nodes if we consider the 4-local case for node~8, but this is obviously not true in this network. Thus node~8 cannot make agreement with other normal nodes. This is because in \cite{su2017reaching}, the more adversarial class of Byzantine nodes is considered.

\begin{figure}[t]
	\centering
	
	\vspace{-10pt}
	\subfigure[\scriptsize{No attack.}]{
		\includegraphics[width=2.8in,height=1.0in]{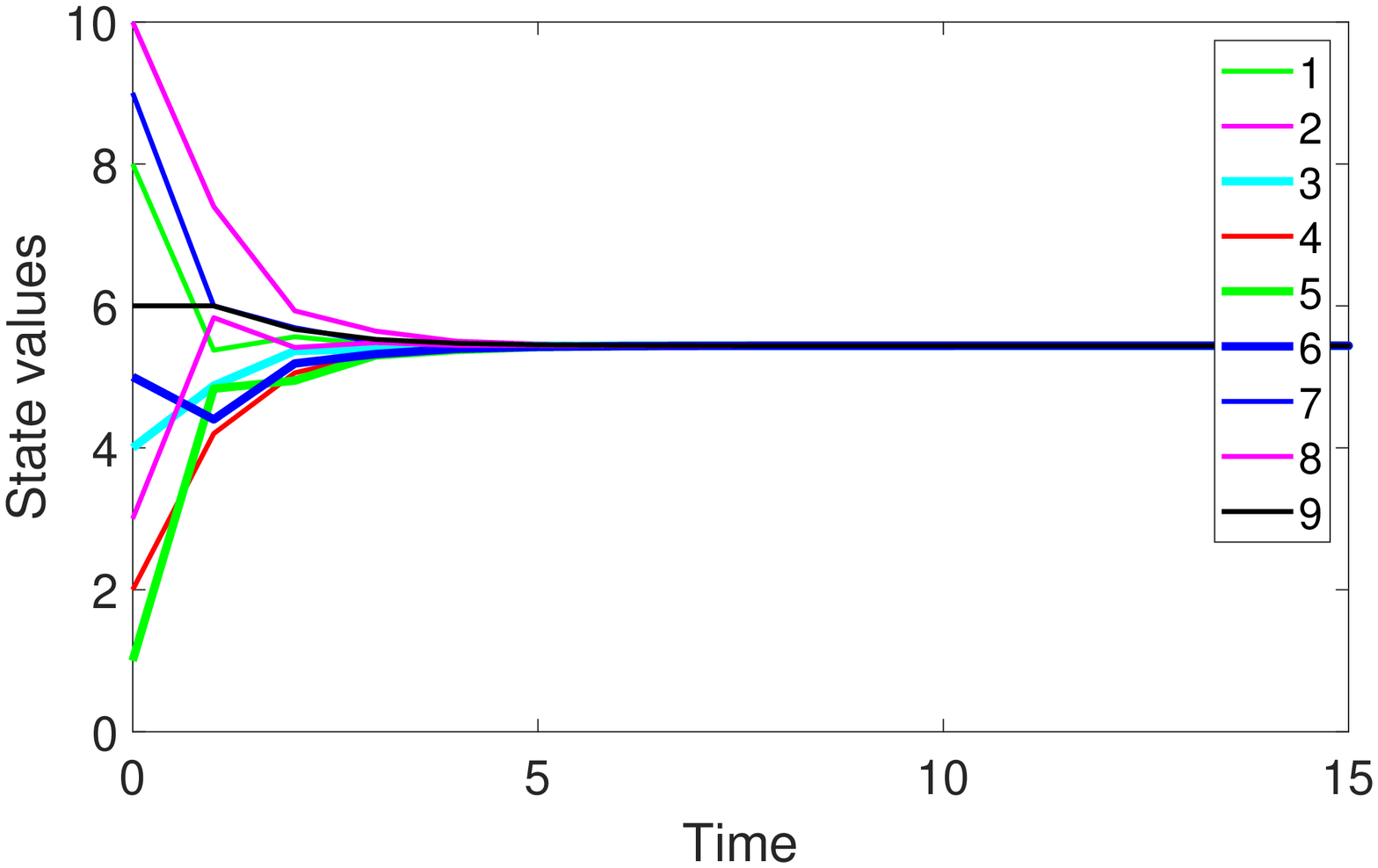}\label{fig4}
	}
	
	\vspace{-12pt}
	\subfigure[\scriptsize{Attack scenario 1.}]{
		\includegraphics[width=2.8in,height=1.0in]{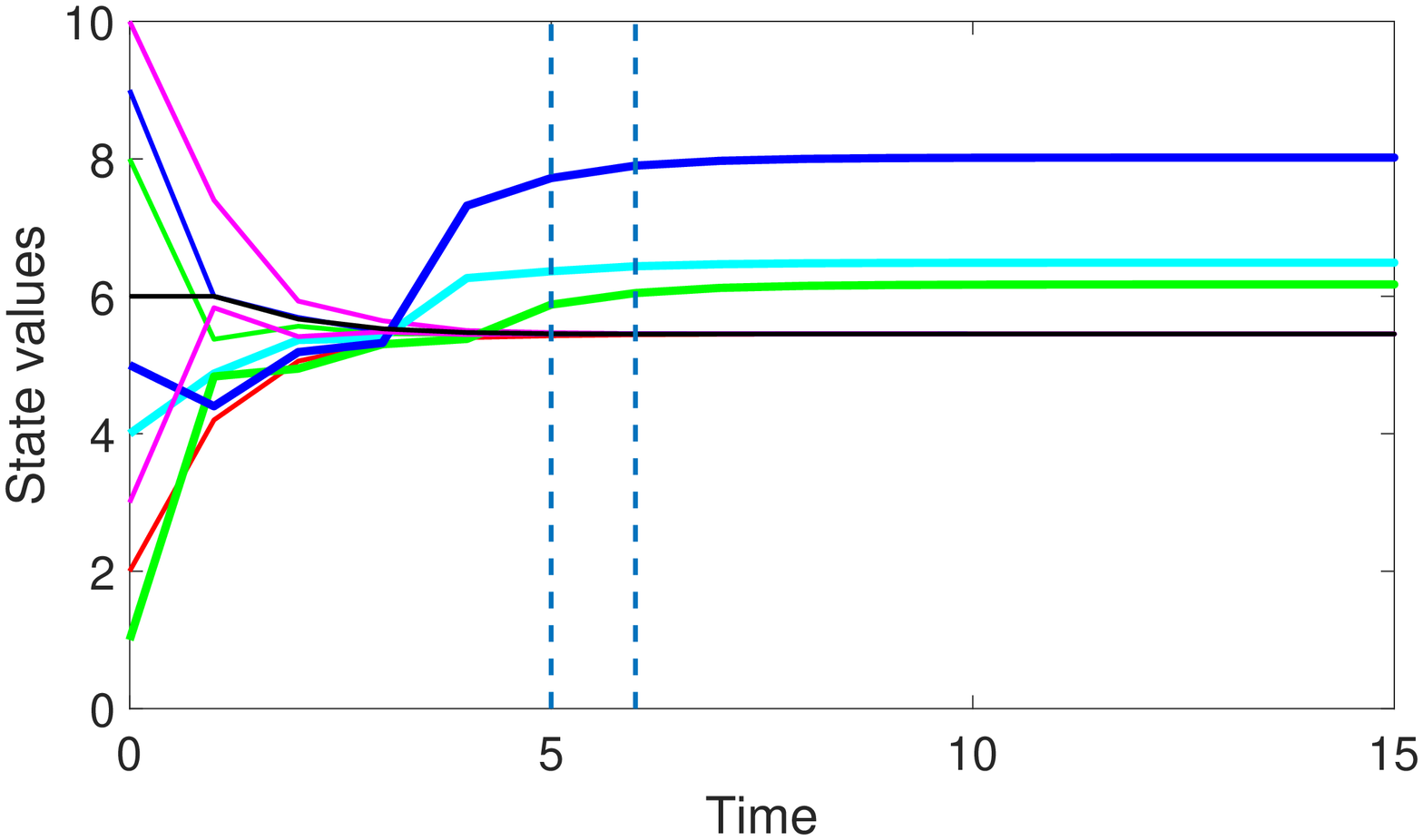}\label{fig5}
	}
	\vspace{-8pt}
	\caption{Scheme 1: Time responses of the states of all nodes.}
\end{figure}

\begin{figure}[t]
	\centering
	
	\vspace{-3pt}
	\includegraphics[width = 1.0in ]{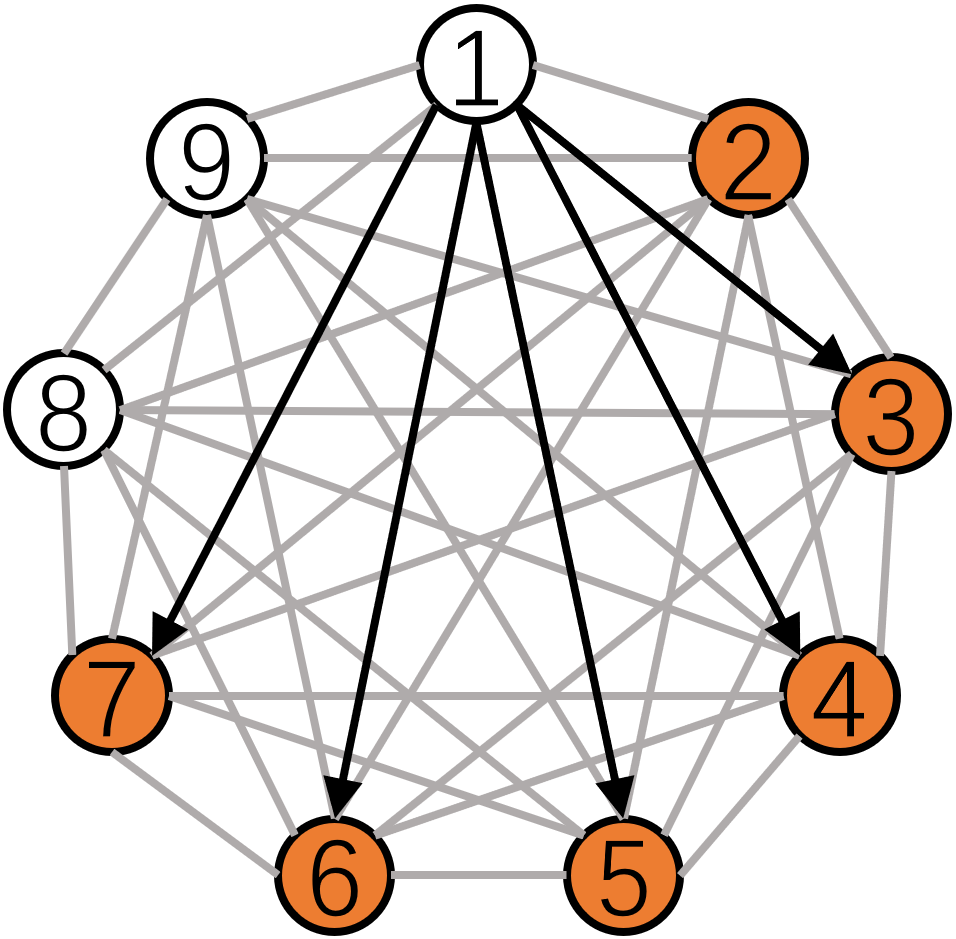}

	\vspace{-4pt}
	\caption{9-node network satisfying the criteria for Scheme~2.}\label{9nodes4malicious}
\end{figure}

\subsection{Application to Large Wireless Sensor Networks}


In this simulation, we create a WSN composed of 100 nodes. At first, we place them at random locations in a $100\times 100$ planar box. Each node can communicate only with the nodes located within the communication radius of $r$. 
Once $r$ is determined, a random geometric network is formed. By increasing communication radius $r$, the network becomes denser and eventually a complete network when $r\geq 122$.
After the network is formed, $f$ nodes are randomly selected to be malicious nodes satisfying our assumptions mentioned before. When we add more malicious nodes in the network, we keep the malicious nodes chosen before and turn normal nodes to new malicious nodes.
Then we apply the proposed schemes and the W-MSR algorithm to the network.

\begin{figure}[]
	\centering
	
	\vspace{-10pt}
	\subfigure[\scriptsize{No attack.}]{
		\includegraphics[width=2.8in,height=1.0in]{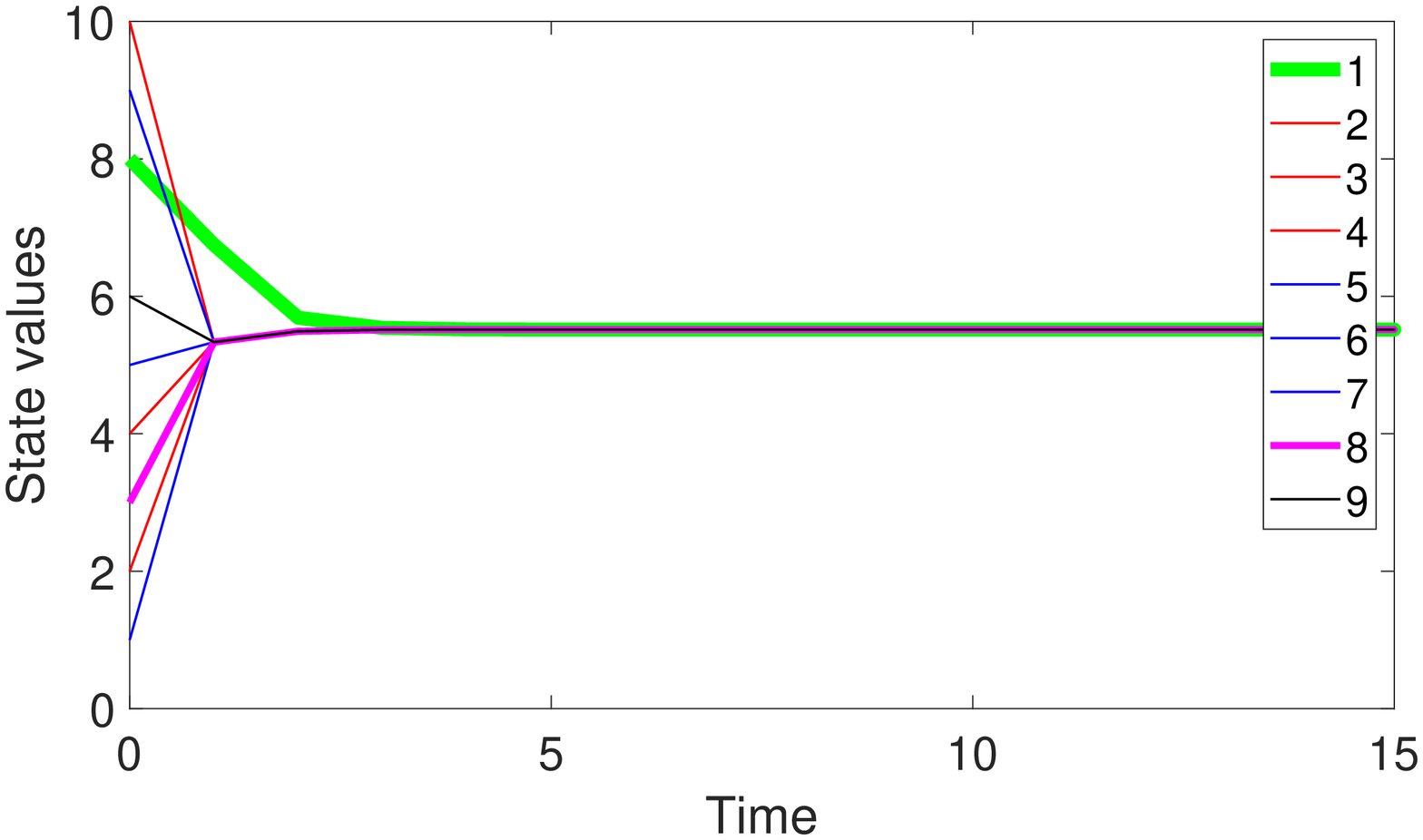}\label{fig14}
	}
	
	\vspace{-12pt}
	\subfigure[\scriptsize{Attack scenario 1.}]{
		\includegraphics[width=2.8in,height=1.0in]{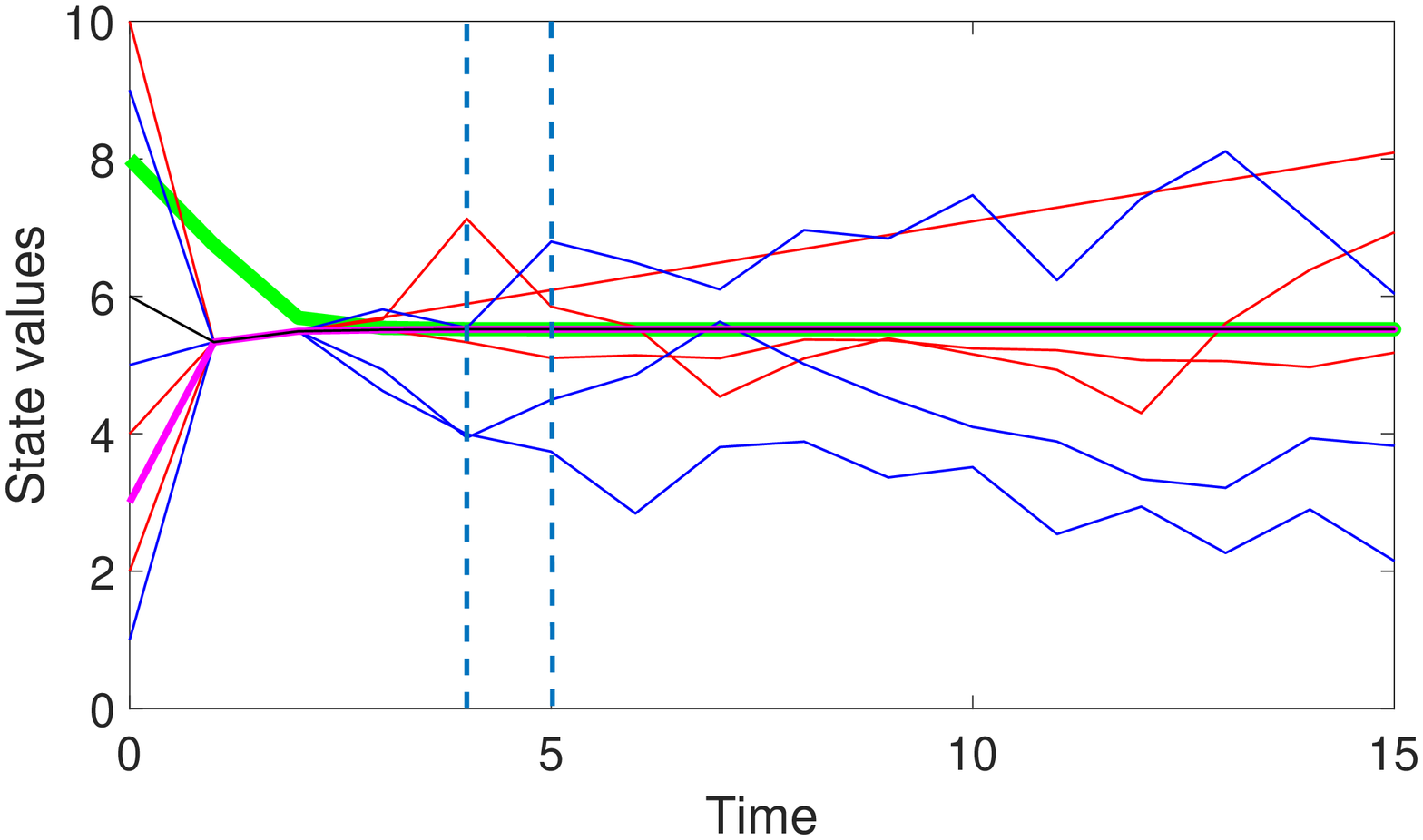}\label{fig15}
	}
	\vspace{-8pt}
	\caption{Scheme 2: Time responses of the states of all nodes.}
\end{figure}

Recall that the malicious nodes can manipulate their own information sets by manipulating either the current value or values from the last time step.
In particular, we consider the following two attack scenarios:
(i)~The first scenario is static in the sense that each malicious node takes a fixed value of $120$. (ii)~The second scenario is more dynamic as each malicious node randomly selects a neighbor and modifies the value from that neighbor arbitrarily and follows the update rule.

Here, we examine how the network connectivity affects the performance of the proposed resilient consensus schemes under the two attack scenarios. Using different values for the number of malicious agents $f$ and the communication radius $r$, we compare the following four algorithms:
(i)~The original consensus algorithm without adversaries,
(ii)~W-MSR algorithm, (iii)~Scheme~1, and (iv)~Scheme~2.
For each $f$ and $r$, we computed the success rate of each algorithm over 20 Monte Carlo runs with randomly chosen initial values of the normal agents in the interval [0,100].
The results of the consensus algorithm without adversaries provide the baseline, indicating when the network becomes connected.

The results under the first attack scenario are presented in Fig.~\ref{fixednode_ak1}. In the plots (a)--(c), we increased the number $f$ of adversaries. Notice that the two proposed schemes are clearly effective against the malicious nodes, and their success rates remain almost the same as the case without any adversaries, where the success rates become 1 around $r=20$. In contrast, the conventional W-MSR degrades in its performance as the number of malicious nodes increases.

The difference between the two proposed schemes becomes more evident under the second attack scenario. 
In Fig.~\ref{fixednode_ak2}, the results are shown as in Fig.~\ref{fixednode_ak1}. It is obvious that Scheme~1 is capable to reach resilient consensus similarly to the previous case under scenario~1.
However, under this scenario, Scheme~2 performs even worse than the W-MSR approach.
On the other hand, an interesting phenomenon can be observed for the case $f=60$ in Fig.~\ref{fixedf60_ak2}, where $f>n/2$ holds. Both Schemes 1 and 2 can guarantee resilient consensus when the network becomes complete while for the W-MSR algorithm, this is not possible. This verifies our analysis before in Section~\ref{Secfors2}.
We highlight again that among the two proposed algorithms, Scheme~2 is fully distributed and more scalable. Thus, there is a tradeoff between the requirements on network connectivities and computation resources.

\begin{figure}[t]
	\centering
	
	\vspace{-10pt}
	\subfigure[\scriptsize{$f=15$.}]{
		\includegraphics[width=3in,height=1.1in]{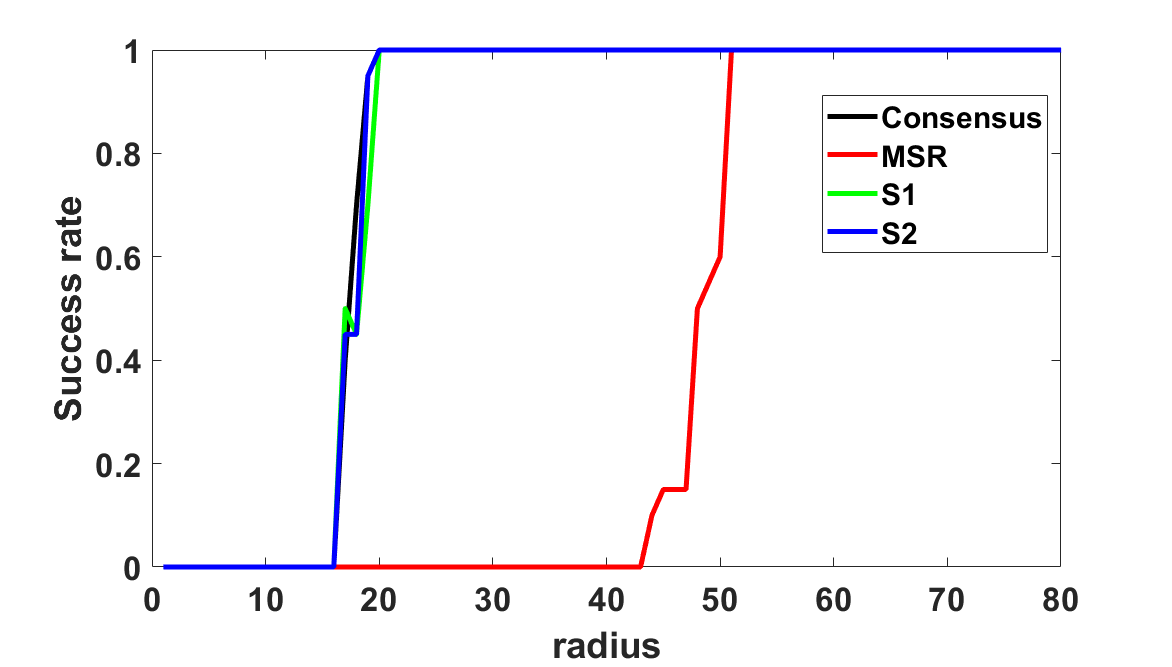}\label{fixedf15.png}
	}
	
	\vspace{-12pt}
	\subfigure[\scriptsize{$f=30$.}]{
		\includegraphics[width=3in,height=1.0in]{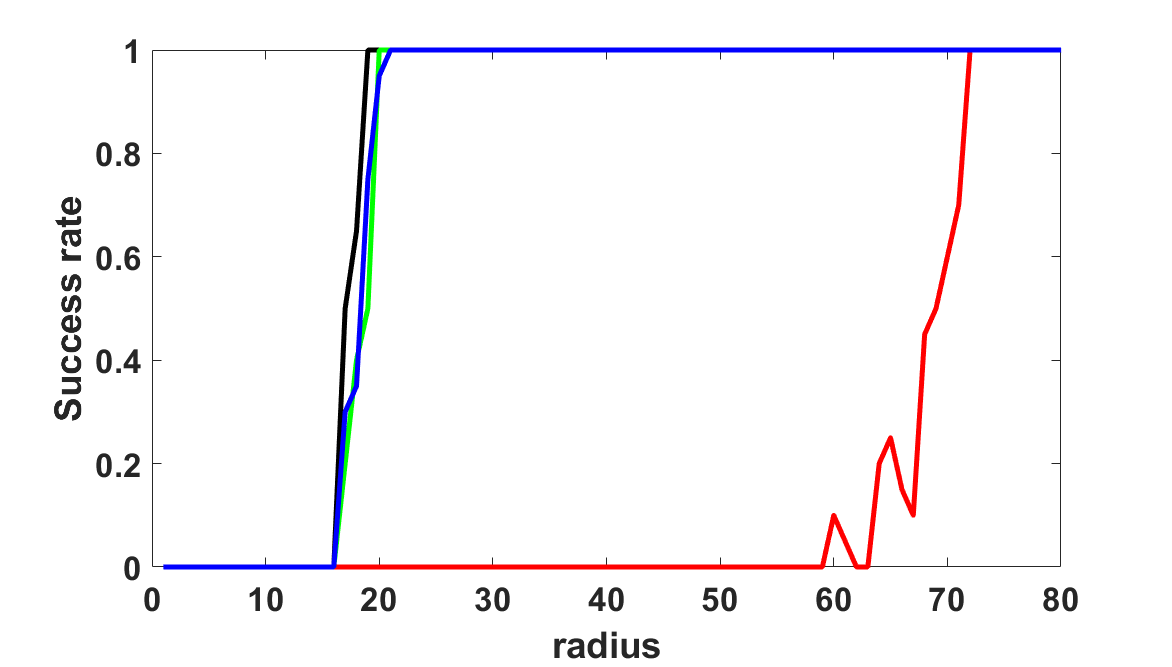}\label{fixedf30.png}
	}

	\vspace{-12pt}
	\subfigure[$f=45$.]{
		\includegraphics[width=3in,height=1.0in]{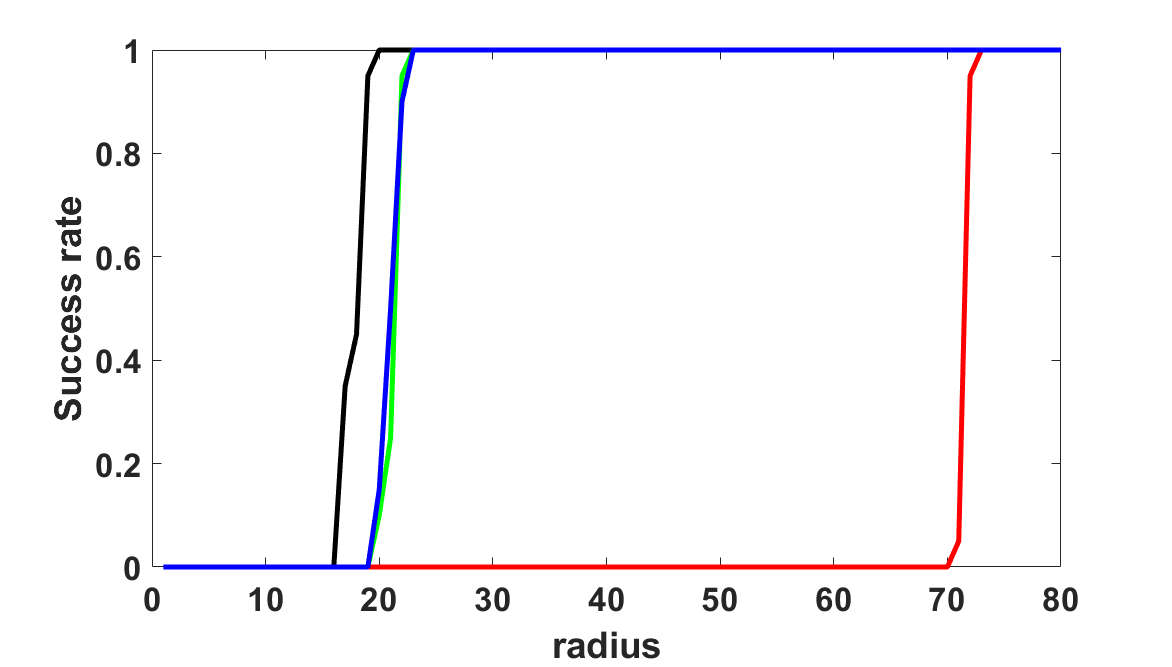}\label{fixedf45}
	}

	\vspace{-8pt}
	
	\caption{Performance of different resilient consensus algorithms under attack scenario 1.}\label{fixednode_ak1}
\end{figure}

We next check that the performance of Scheme~2 can be significantly improved by increasing the number of edges in the network.
Earlier in Section~\ref{Secfors2}, through example graphs in Fig.~\ref{graph9fors2}, we have seen such a property more analytically. Here, to increase edges, we introduce 16 additional relay nodes (\cite{zhang2007fault}) as discussed in Section \ref{discuss2}. Such a node has strong communication capabilities; it receives data from the nodes within their communication ranges of radius $r$ and then simply sends out, or relays, the received data to nodes within its own communication radius $r_\text{relay} = r+27$.
Relay nodes are located at the coordinates $(20x,20y), x, y=1,2,3,4$. In our setting, such nodes only relay information received from general nodes and not from other relay nodes; also they do not conduct any detection nor consensus algorithm.
These nodes bring in the same effects as introducing more directed edges in the network. In Fig.~\ref{fixednode_ak2}, the success rates of Scheme 2 and the W-MSR algorithm are indicated, respectively, by the blue and red dashed lines. One can observe that Scheme 2 performs almost the same as Scheme 1 and much better than the W-MSR algorithm especially when $f$ grows.

\begin{figure}[t]
	\centering
	
	\vspace{-10pt}
	\subfigure[\scriptsize{$f=15$.}]{
		\includegraphics[width=3in,height=1.1in]{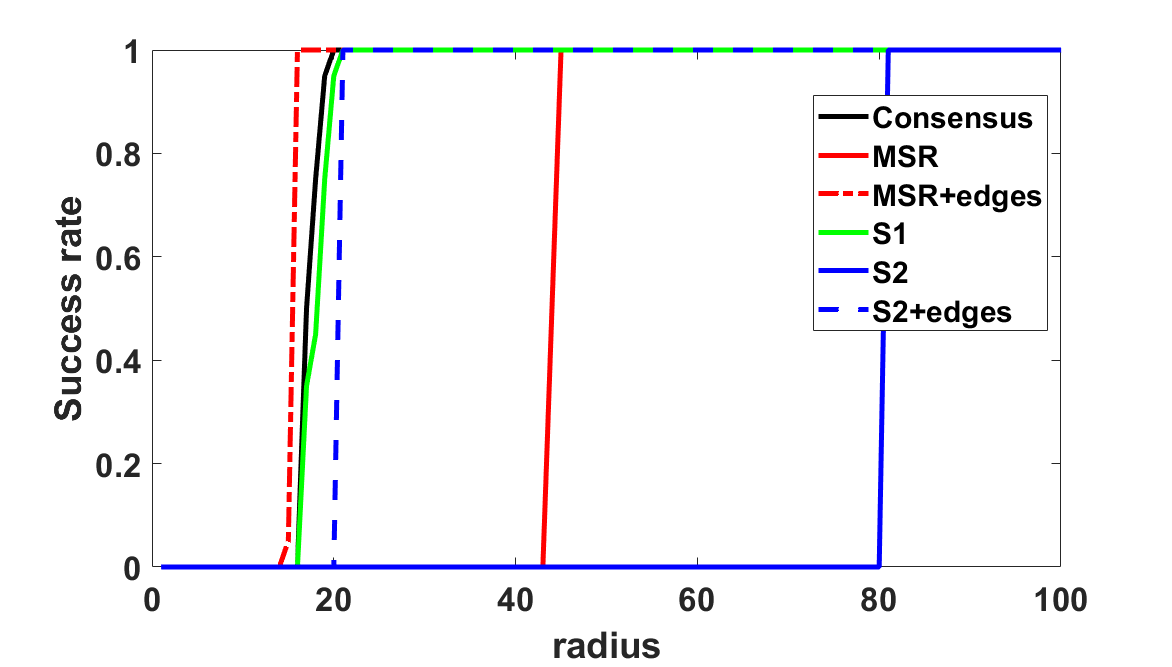}\label{fixedf15_ak2}
	}
	
	\vspace{-12pt}
	\subfigure[\scriptsize{$f=30$.}]{
		\includegraphics[width=3in,height=1.0in]{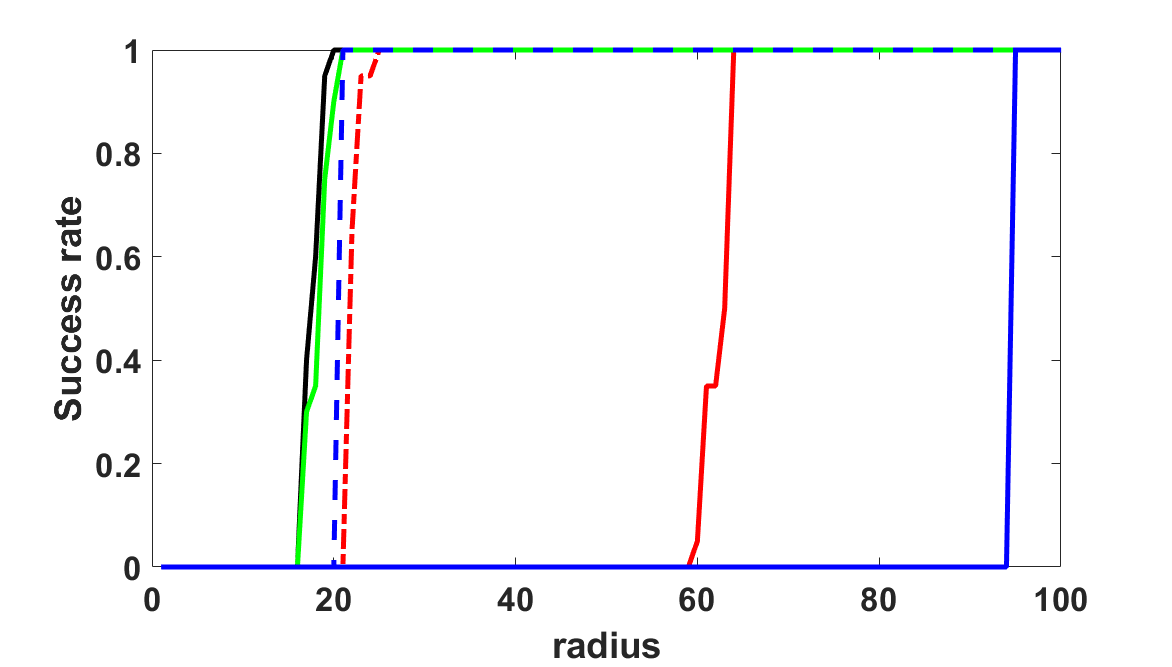}\label{fixedf30_ak2}
	}

	\vspace{-12pt}
	\subfigure[$f=45$.]{
		\includegraphics[width=3in,height=1.0in]{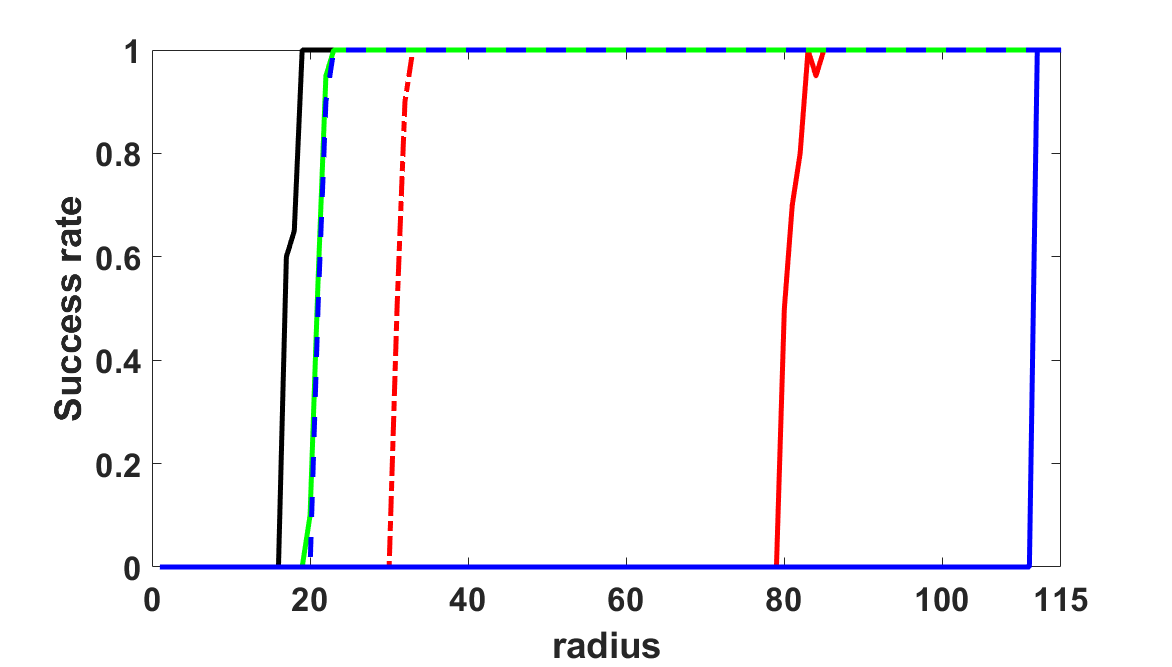}\label{fixedf45_ak2}
	}

	\vspace{-12pt}
	\subfigure[\scriptsize{$f=60$.}]{
		\includegraphics[width=3in,height=1.0in]{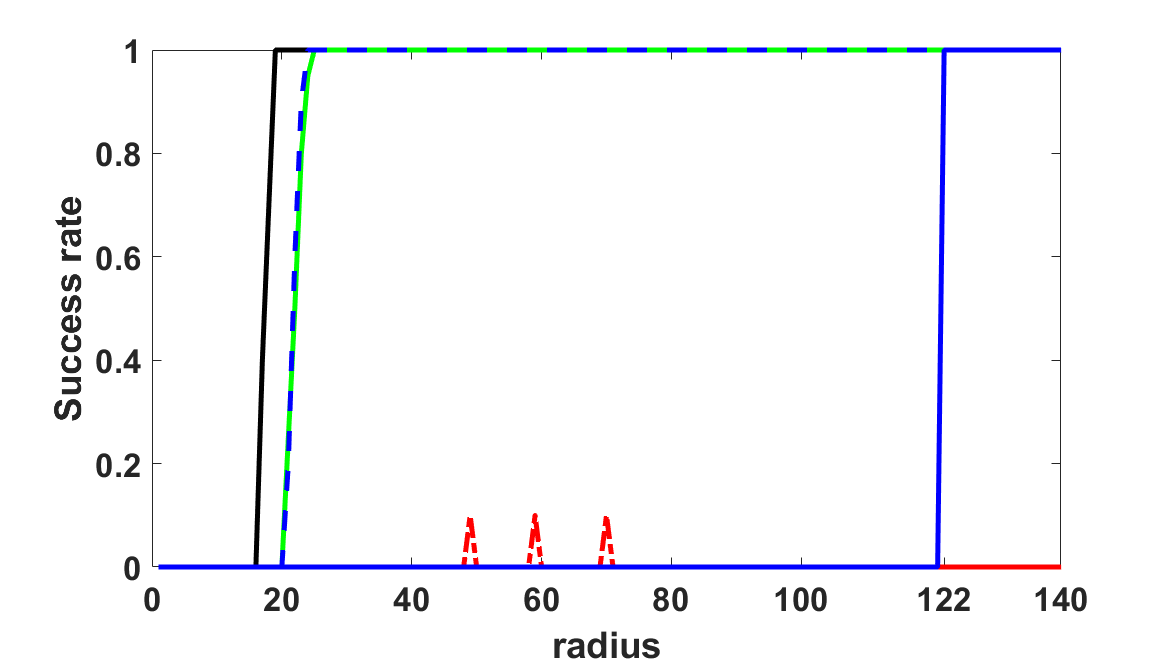}\label{fixedf60_ak2}
	}
	
	\vspace{-8pt}
	
	\caption{Performance of different resilient consensus algorithms under attack scenario 2.}\label{fixednode_ak2}
\end{figure}

%
%
%
%
%
%
%

\section{Conclusions}

In this paper, we have designed two novel distributed detection
schemes to solve the resilient consensus problem. The key features
of the schemes lie in the assumption on the adversaries based on
the malicious agent model and the use of two-hop communication among
the agents. We have clarified that the levels of network connectivities  
for both schemes can be more sparse compared to conventional approaches.
In the two schemes, the normal agents perform as detectors by monitoring
the behaviors of their neighbors, but they are different
in terms of distributed computation capabilities and the required
network connectivities. We have presented their properties through
extensive numerical examples. 

In future research, one possible direction is to apply two-hop communication techniques to enhance resilience in other
multi-agent coordination problems. 
Also, the extension of our results to multi-hop communication with more than two-hop is left for further work. We believe that if each node has more information about the multi-hop neighbors, we can reduce the requirement on the network structures and accelerate the consensus process.

\textit{Acknowledgement} The authors wish to thank Xavier D\'{e}fago for the helpful discussions and thank the reviewers for their valuable comments and suggestions.

\appendix
\section{Proof of Lemma \ref{lemma1} } \label{appendixa}

	\emph{Necessity:} 
	We show by contradiction. Suppose that malicious nodes $j$ and $l$ have no common normal neighbors. Then there is no normal node knowing if node $j$ changes $x_l^{(j)}[k-1|k]$ in $\Phi_j[k]$ or not.
	Node $l$ can make similar manipulations. Thus nodes $j$ and $l$ cannot be detected.
	
	\emph{Sufficiency:} Without loss of generality, consider the pair of malicious nodes $j$ and $l$ and their common neighbor normal node $i$. Three nodes form a triangular subgraph.
	Node $j$ will be detected by node $i$ if it changes $x_j^{(j)}[k|k]$ in $ \Phi_j[k]$ by breaking the update rule or changes $x_i^{(j)}[k-1|k]$ in $ \Phi_j[k]$. Moreover, it will still be detected by node $i$ if it changes $x_j^{(j)}[k-1|k],x_l^{(j)}[k-1|k]$ in $\Phi_j[k]$ since node $i$ records $x_j^{(j)}[k-1|k-1]$ in $\Phi_j[k-1]$ and $x_l^{(l)}[k-1|k-1]$ in $\Phi_l[k-1]$. We can do the same analysis for malicious node $l$. For the case that node $j$ has independent normal neighbors other than the common neighbors of nodes $j$ and $l$, it will be detected by the normal neighbors if the corresponding values from the normal neighbors in $\Phi_j[k]$ are changed. Based on the assumption, any manipulation on the information sets of the malicious nodes under the common neighbor condition will be detected. \hfill  $\blacksquare$

\section{Proof of Theorem \ref{detect}} \label{appendixb}        
           
            \emph{(a) Necessity:} 
           	Suppose that there is a pair of neighboring nodes, but they have at most $f-2$ two-hop paths connecting them, i.e., they have $f-2$ common neighbors. In this case, we can take the pair of nodes as well as all of their common neighbors to be malicious. Then no normal node is a common neighbor of the pair of malicious nodes. According to Lemma~\ref{lemma1}, the cooperation between the malicious nodes will not be detected by any of the normal ones in the network.
           	
           	\emph{Sufficiency:} 
           	By the assumption on the connectivity of $\mathcal{G}$, every malicious node has at least one neighbor. We consider the following two cases separately: A malicious node $i$ has (i) only normal neighbors and (ii) one or more malicious neighbors.
           	We will show that in each case, the malicious node $i$ will be detected by one of its normal neighbors. At that point, by Assumption \ref{broadcast}, all normal nodes will be informed of node $i$ being malicious. 
           	
           	(i) When malicious node $i$ has only normal neighboring nodes, denote them by nodes 1, 2, \dots, $d_i$. Node $i$ will be detected by any node $j$ among its neighbors if it changes $x_i^{(i)}[k|k] $ in $ \Phi_i[k]$ by breaking the update rule \eqref{updaterule}. Moreover, it can manipulate $x_i^{(i)}[k|k] $ in $ \Phi_i[k]$ through changing some of the data in $\Phi_i[k]$, which still satisfies the update rule. According to the assumptions, it can change or delete its own previous value $x_i^{(i)}[k-1|k]$ or any of the values of $x_1^{(i)}[k-1|k],x_2^{(i)}[k-1|k],...,x_{d_i}^{(i)}[k-1|k]$. Any such changes in the information set will be detected by its normal neighbors whose values are changed or deleted since the same information set is broadcast.
           	
           	(ii) When malicious node $i$ has one or more malicious nodes as its neighbors, we prove that the condition in Lemma~\ref{lemma1} is satisfied. Take one malicious neighbor of node $i$. By assumption, this pair of nodes have at least $f-1$ common neighbors. Since in the $f$-total model, there are at most $f$ malicious nodes in the network, at most $f-2$ neighbors can be malicious. Hence, among the $f-1$ common neighbors, one or more must be normal. Thus, the condition in Lemma~\ref{lemma1} is satisfied. This then implies 
           	that the pair of malicious nodes will be detected by this normal node.
           	
           	Therefore, in both cases (i) and (ii), malicious nodes are detected. 
           	
           	\emph{(b)} Malicious nodes will be detected immediately once they misbehave. Thus misbehaviors of malicious nodes cannot affect normal nodes since normal nodes exclude values from detected malicious nodes. The safety condition is thus guaranteed. Furthermore, if $\mathcal{G}$ is ($f+1$)-connected, even after removing $f$ malicious nodes from the network, the normal nodes are still connected. Therefore, resilient consensus is achieved.  \hfill  $\blacksquare$

\section{Proof of Proposition \ref{minimumgraph}} \label{appendixc}

	Consider the case where the graph is constructed from a complete graph $\mathcal{K}_n$ with edge $(j,i)$ removed. Such a graph has ($n-1$)-connected rooted spanning trees. Moreover, nodes other than nodes $i$ and $j$ are full access nodes that can, by Corollary \ref{fullaccessnode}, detect any malicious nodes in the network. Thus we only focus on nodes $i$ and $j$. Note that there are $n-2$ directed two-hop paths from $j$ to $i$. For node $i$ to obtain the true value of node $j$, the majority of these paths should contain normal nodes as $l$ in Theorem \ref{detect2}. For the network to tolerate more malicious nodes, the extreme case is that node $j$ is malicious. Thus we have
	$(n-2)-(f-1)  >  \left \lfloor  \frac{n-2}{2} \right \rfloor   \Longleftrightarrow n-\left\lfloor \frac{n}{2} \right\rfloor  >   f$. 
	Note that $-\left\lfloor n/2 \right\rfloor=\left\lceil -n/2 \right\rceil$ and $n=\left\lceil n \right\rceil$, and hence $\left\lceil n/2 \right\rceil>f$.
	When $n$ is even, we have 
	$\frac{n}{2} >  f 	\Longleftrightarrow	n  >  2f $.
	When $n$ is odd, we have 
	$\frac{n+1}{2} >  f 	\Longleftrightarrow	n \geq 2f+1 >2f$.
	After $f$ malicious nodes are detected and removed, where $f\leq n-2$, the graph has at least one rooted spanning tree, and thus resilient consensus is guaranteed. \hfill  $\blacksquare$

\fontsize{8pt}{8.5pt}\selectfont

\end{document}